\label{key}
\documentclass[10pt,journal,compsoc]{IEEEtran}

\usepackage{amsmath}

\usepackage{cite}
\usepackage{amsmath,amssymb,amsfonts}
\allowdisplaybreaks

\usepackage{array}
\usepackage{graphicx}
\usepackage{stfloats}
\usepackage{amsthm}
\usepackage{caption}
\usepackage{subfigure}
\usepackage{url}
\usepackage{verbatim}  
\usepackage{textcomp}
\usepackage{xcolor}
\usepackage{color}
\usepackage{enumitem}
\usepackage{algorithmicx}
\usepackage[noend,ruled, linesnumbered, vlined]{algorithm2e}
\usepackage{algpseudocode}
\usepackage[absolute,overlay]{textpos}

\allowdisplaybreaks[4]
\def\BibTeX{{\rm B\kern-.05em{\sc i\kern-.025em b}\kern-.08em
    T\kern-.1667em\lower.7ex\hbox{E}\kern-.125emX}}
\newenvironment{sequation}{\begin{equation}\small}{\end{equation}}

\newtheorem{theorem}{\textbf{Theorem}}
\newtheorem{remark}{\textbf{Remark}}
\newtheorem{definition}{\textbf{Definition}}

\usepackage{multirow}
\usepackage{bm}
\usepackage{booktabs}

\usepackage{epstopdf}
\usepackage{cases}
\usepackage{makecell,multirow,diagbox}
\usepackage{stfloats}
\usepackage{comment}
\usepackage{array}

\usepackage{xr}

\makeatletter
\newcommand*{\addFileDependency}[1]{
  \typeout{(#1)}
  \@addtofilelist{#1}
  \IfFileExists{#1}{}{\typeout{No file #1.}}
}
\makeatother

\newcommand*{\myexternaldocument}[1]{
    \externaldocument{#1}
    \addFileDependency{#1.tex}
    \addFileDependency{#1.aux}
}

\myexternaldocument{Appendices}

\definecolor{b}{rgb}{0.0, 0, 0}
\definecolor{k}{rgb}{0, 0, 0}
\definecolor{c}{rgb}{0, 0, 0}

\def\BibTeX{{\rm B\kern-.05em{\sc i\kern-.025em b}\kern-.08em
		T\kern-.1667em\lower.7ex\hbox{E}\kern-.125emX}}

\begin{document}

\title{Online Collaborative Resource Allocation and Task Offloading for Multi-access Edge Computing}

\author{Geng~Sun,~\IEEEmembership{Senior Member,~IEEE},
        Minghua Yuan,
        Zemin Sun,~\IEEEmembership{Member,~IEEE}, 
        Jiacheng Wang, \\
        Hongyang~Du,
        Dusit Niyato,~\IEEEmembership{Fellow,~IEEE},
        Zhu Han,~\IEEEmembership{Fellow,~IEEE}, and 
        Dong In Kim,~\IEEEmembership{Fellow,~IEEE}
\thanks{This study is supported in part by the National Natural Science Foundation of China (62272194, 62471200), and in part by the Science and Technology Development Plan Project of Jilin Province (20230201087GX). (\textit{Corresponding author: Zemin Sun)}}
 \IEEEcompsocitemizethanks{
     \IEEEcompsocthanksitem Geng Sun is with the College of Computer Science and Technology, Jilin University, Changchun 130012, China, and also with the College of Computing and Data Science, Nanyang Technological University, Singapore 639798 (e-mail: sungeng@jlu.edu.cn).
     \IEEEcompsocthanksitem Minghua Yuan is with the College of Software Engineering, Jilin University, Changchun 130012, China (e-mail: yuanmh22@mails.jlu.edu.cn).
    \IEEEcompsocthanksitem Zemin Sun is with the College of Computer Science and Technology, Jilin University, Changchun 130012, China (e-mail: sunzemin@jlu.edu.cn).
     \IEEEcompsocthanksitem Jiacheng Wang is with the College of Computing and Data Science, Nanyang Technological University, Singapore 639798 (e-mail: jiacheng.wang@ntu.edu.sg).
    \IEEEcompsocthanksitem Hongyang Du is with the Department of Electrical and Electronic Engineering, the University of Hong Kong, Hong Kong, SAR, China. (e-mail: duhy@eee.hku.hk).
    \IEEEcompsocthanksitem Dusit Niyato is with the College of Computing and Data Science, Nanyang Technological University, Singapore 639798 (e-mail: dniyato@ntu.edu.sg).
     \IEEEcompsocthanksitem Zhu Han is with the Department of Electrical and Computer Engineering at the University of Houston, Houston, TX 77004 USA, and also with the Department of Computer Science and Engineering, Kyung Hee University, Seoul, South Korea, 446-701 (e-mail:  zhan2@uh.edu).
    \IEEEcompsocthanksitem Dong In Kim is with the Department of Electrical and Computer Engineering, Sungkyunkwan University, Suwon 16419, South Korea (email: dongin@skku.edu).
    }}


\IEEEtitleabstractindextext{%
\begin{abstract}	
\par Multi-access edge computing (MEC) is emerging as a promising paradigm to provide flexible computing services close to user devices (UDs). However, meeting the computation-hungry and delay-sensitive demands of UDs faces several challenges, including the resource constraints of MEC servers, inherent dynamic and complex features in the MEC system, and difficulty in dealing with the time-coupled and decision-coupled optimization. In this work, we first present an edge-cloud collaborative MEC architecture, where the MEC servers and cloud collaboratively provide offloading services for UDs. Moreover, we formulate an
energy-efficient and delay-aware optimization problem
(EEDAOP) to minimize the energy consumption of UDs under the  constraints of task deadlines and long-term queuing delays. Since the problem is proved to be non-convex mixed integer nonlinear programming (MINLP), we propose an online joint communication resource allocation and task offloading approach (OJCTA). Specifically, we transform EEDAOP into a real-time optimization problem by employing the Lyapunov optimization framework. Then, to solve the real-time optimization problem, we propose a communication resource allocation and task offloading optimization method by employing the Tammer decomposition mechanism, convex optimization method, bilateral matching mechanism, and dependent rounding method. Simulation results demonstrate that the proposed OJCTA can achieve superior system performance compared to the benchmark approaches.
\end{abstract}
	
}

\maketitle
\IEEEdisplaynontitleabstractindextext
\IEEEpeerreviewmaketitle

%
%

\section{Introduction}
\label{sec_introduction}

\par  \IEEEPARstart{W}{ith} the proliferation of the Internet of Things (IoT) devices and the development of the sixth generation (6G) technology, various new applications emerging, such as the online game, virtual reality (VR), augmented reality (AR), mixed reality (MR), and healthcare monitoring. These applications are typically delay-sensitive and computation-hungry, imposing strict requirements on computing capabilities. However, due to the physical constraints of the user devices (UDs), they are typically equipped with limited computing and energy resources, making it difficult to meet the demanding requirements of applications. To overcome these challenges, multi-access edge computing (MEC) is considered a promising technology by migrating cloud computing capabilities in close proximate to UDs. By offloading tasks to MEC servers, the delay-sensitive and computation-hungry tasks can be processed in an energy-efficient and cost-effective way. Accordingly, the computing capability and battery life of the UDs can be extended.

\par Despite the abovementioned benefits, several significant challenges should be addressed to fully exploit the potential of MEC for efficient task offloading. Different from traditional cloud computing system, a MEC system is constrained by limited resources. Specifically, the UDs such as the wearable, handheld, and VR/AR/MR equipment, are typically equipped with limited battery capacity to enhance portability. However, once the battery is depleted, the devices are unable to complete their tasks. This is particularly critical for healthcare applications that require constant operation to monitor physiological data. Moreover, due to the physical and spectrum limitations in wireless communications, MEC servers have limited communication resources. Additionally, unlike centralized cloud, MEC servers are also equipped with limited computational capabilities. However, the explosive growth of computation-intensive and delay-sensitive tasks puts stringent requirements on MEC servers. As a result, the strict requirements of UDs and the limited resources of MEC servers present challenges in designing efficient task offloading and resource management strategies. 

\par  An MEC system also exhibits greater complex features compared to conventional wireless systems. Specifically, the MEC system is dynamic, due to fluctuating network conditions, varying UD mobility patterns, unpredictable task arrivals, and continuous task processing. Additionally, UDs have limited energy, MEC servers have constrained resources, and most tasks are delay-sensitive and computation-intensive. Thus, incorporating these complex features into an optimization model is challenging, as it requires accurately capturing the complexity and variability of the MEC system. 

\par The complexity of the optimization model introduces additional challenges in solving the corresponding optimization problem. First, the dynamic nature of MEC systems results in time-coupled optimization problems. In particular, the delay sensitivity of tasks necessities real-time decision making, which requires complete information about the system. However, the dynamics of MEC systems makes it difficult to predict the UD demands and resource availability. Moreover, different decisions in MEC systems are interdependent. For example, the task offloading decisions of UDs and the resource allocation decisions of MEC servers are coupled and jointly affect both immediate performance and the future system states. Consequently, this complexity and interdependence complicate the design of efficient approaches for solving the optimization problem.


\par Addressing the aforementioned challenges requires efficient optimization of resource allocation and task offloading. To this end, the work presents an online collaborative task offloading and resource allocation approach for the MEC system. Our main contributions are outlined as follows.

\begin{itemize}
	\item \textit{\textbf{System Architecture.}} We employ a hierarchical edge-cloud collaborative MEC architecture, composed of a UDs layer, an MEC layer, and a cloud layer. Specifically, tasks generated at the UD layer can be executed locally, or offloaded to the MEC servers. Then, at the MEC layer, the MEC servers can select to process the tasks or offload them to the cloud layer. The cloud server provides support by collaboratively managing excess workloads from the MEC servers.

	\item \textit{\textbf{Problem Formulation.}} Given the energy constraints of the UDs and the delay sensitivity of tasks, we formulate an energy-efficient and delay-aware optimization problem (EEDAOP). A goal of EEDAOP is to minimize the average UD energy consumption while meeting the task deadlines of UDs and long-term queuing delay constraints at the edge. However, EEDAOP is proved to be an NP-hard mixed integer nonlinear programming (MINLP) problem.
	
    \item \textit{\textbf{Algorithm Design.}} To solve the formulated EEDAOP, we propose an online communication resource allocation and task offloading approach (OJCTA). First, to deal with the dependence of the problem, we transform EEDAOP into a real-time optimization problem by applying the Lyapunov optimization framework. Then, to solve the real-time optimization problem, we propose a communication resource allocation and task offloading optimization method that employs the Tammer decomposition mechanism, convex optimization method, bilateral matching mechanism, and dependent rounding method.

    \item \textit{\textbf{Validation.}} The performance of the proposed OJCTA is analyzed and demonstrated through theoretical analysis and simulation. Specifically, we theoretically prove that the proposed OJCTA has a polynomial worst-case computational complexity. Moreover, the simulation results demonstrate OJCTA outperforms other comparative approaches.
\end{itemize}

\par The rest of the paper is organized as follows. Section \ref{sec_related work} provides an overview of the related work. Section \ref{sec_model} describes the relevant models. In Section \ref{sec_pro_analysis}, we outline the problem formulation along with a detailed analysis. Section \ref{sec_ProblemTransfor} details the proposed OJCTA approach. Section \ref{sec_simulation} gives the simulation results and the analysis. Finally, Section \ref{sec_conclusion} presents the conclusion of this work.

%
%
\section{Related work}
\label{sec_related work}







\subsection{MEC Architecture}

\par MEC is emerging as a promising architecture to extend the computing capabilities of UDs, and has attracted increasing research attention. From the perspective of network architecture, the existing studies can be primarily categorized into single-MEC-server architectures and multiple-MEC-server architectures.

\par For single MEC architectures, Tao \emph{et al.}~\cite{TaoM2021} considered a multi-user scenario with a single MEC server, allowing users to offload some of their computational tasks to the edge server while processing the rest locally. Xiao \emph{et al.}
\cite{xiao2022multi} studied a device to device-aided MEC system, in which a remote cloud and an edge server are deployed to offer computing capabilities for the mobile users. Moreover, Wang \emph{et al.}~\cite{WJY2023} explored the resource management in a non-orthogonal multiple access (NOMA)-MEC system, where users offload the tasks to the MEC server by using the NOMA scheme. However, the single MEC server architecture is constrained by the limited coverage and resources, making it less suitable for the large-scale or dense-populated scenarios.

\par To address the limitations of the single-MEC-server-architecture, researchers explored multiple-MEC-server architecture. For example, Ndikumana \emph{et al.}~\cite{NA2020} presented a collaborative MEC architecture, where multiple MEC servers are grouped into a cluster to collaboratively process the computation tasks of users. 
Yi \emph{et al.}~\cite{YCY2021} proposed a workload re-allocation approach by stimulating different MEC servers to participate in server collaboration for workload exchange. However, compared to cloud computing, the resources of MEC servers are inherently constrained, making it challenging to meet the computational intensive and latency sensitive requirements of users in densely populated scenarios, especially during peak usage periods or when multiple users simultaneously offload tasks.

\par The existing architectures are inadequate in scenarios characterized by high user density and significant computational resource demands. To address this, we present a hierarchical edge-cloud collaborative MEC architecture, where the cloud server assists MEC servers in task processing. This architecture effectively combines the powerful computational capability of the cloud server with the low-latency benefit of edge servers.

\subsection{Optimization Metrics of MEC System}

\par Recent studies have explored various aspects of metric optimization for MEC systems, with the primary focuses on task offloading and resource allocation.

\par Task offloading is a critical technique in MEC systems and has been extensively investigated by recent studies. For example, Chen \emph{et al.}~\cite{Chen2024Amulti} focused on minimizing the maximum delay by optimizing the multi-hop task offloading for MEC-assisted vehicular system. To minimize the wireless energy transmission consumption, Bolourian \emph{et al.}~\cite{Bolourian2023Energy} proposed a harvest-then-offload approach for the wireless-powered MEC system. Moreover, Gao \emph{et al.}~\cite{Gao2023Game} aimed to minimize the maximum total task delay by optimizing the task offloading decisions in wireless powered MEC system.

\par Some studies further investigated joint resource allocation and task offloading to fully leverage the advantages of MEC. For example, Hoang \emph{et al.} \cite{hoang2024joint} formulated an utility maximization problem that integrates completion delay, energy consumption, and achievable rate via jointly optimizing the radio resource allocation and task offloading in stochastic MEC systems. Diamanti \emph{et al.} \cite{Diamanti2024Delay} aimed to minimize the maximum experienced delay of users by optimizing both task assignment ratios and the allocated computing resources. Mei \emph{et al.} \cite{Mei2024Through} presented a joint optimization problem of task offloading and resource allocation to maximize the throughput of the MEC system. Moreover, Liu \emph{et al.}~\cite{liu2023joint} aimed to achieve latency reduction through the optimization of access point selection, task offloading, and resource allocation in a heterogeneous MEC environment.

\par The abovementioned works mainly optimize the performance metrics of delay or energy consumption for the MEC system. However, extensive task offloading could result in workload accumulation at the MEC servers, which in turn causes queuing delays at the edge. The queuing delay could significantly degrade the system performance, especially during peak usage periods or in high-demand scenarios. In contrast to the previous studies, our work addresses an energy-efficient and delay-aware optimization problem to minimize UD energy consumption under the delay constraints. The delay constraints include task deadlines for UDs and long-term queuing delay limits at the edge, which ensure task processing timeliness and stability at MEC servers, respectively. The constraints are essential for applications such as the 
telemedicine and online gaming where both instantaneous response and sustained performance stability should be guaranteed~\cite{jiang2023energy}.

\begin{figure}[t]
\captionsetup{justification=centering}
\includegraphics[width =3.5in]{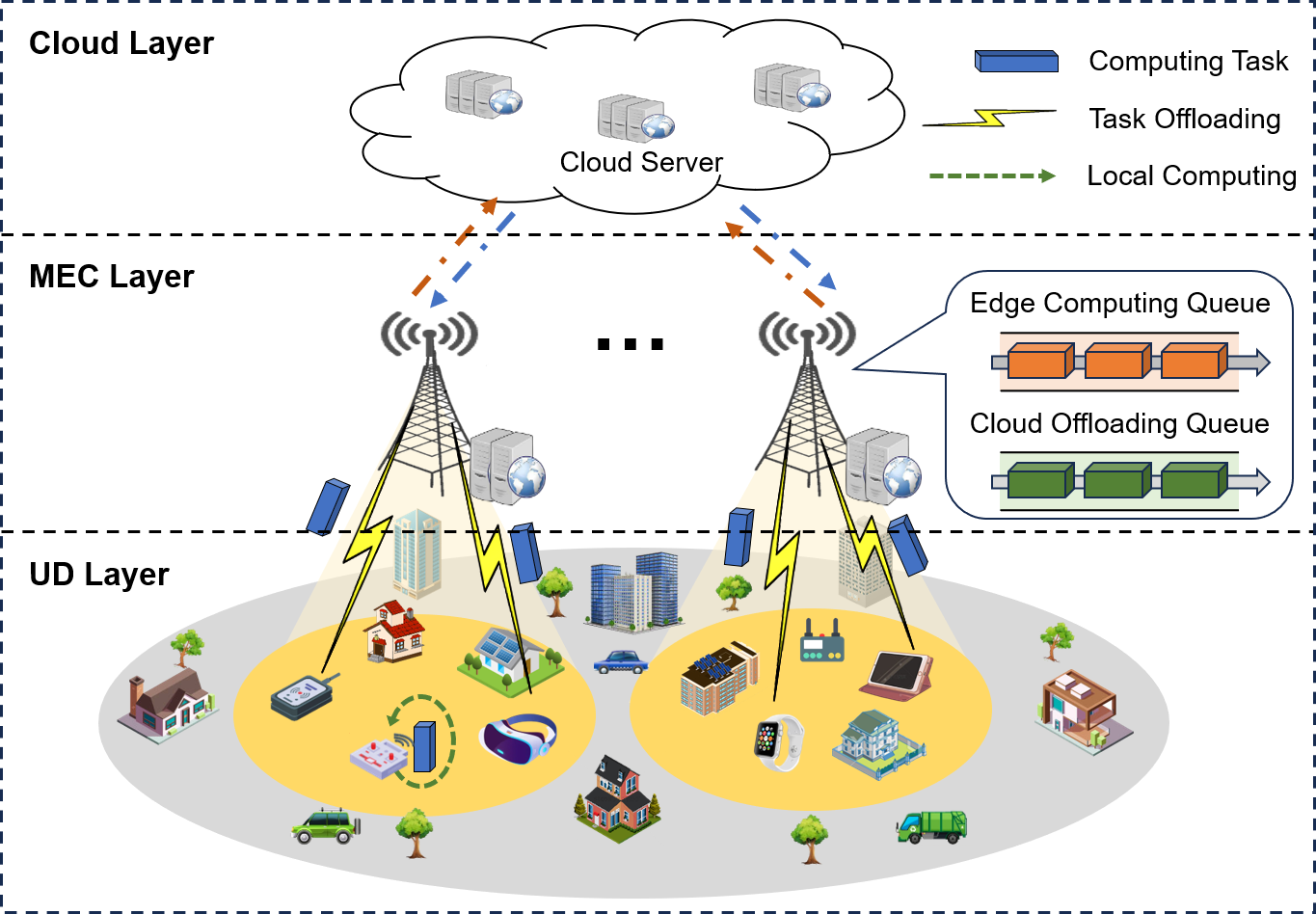}
\caption{The proposed three-layer collaborative MEC system architecture.}
\label{fig_systemModel}
\end{figure}

\subsection{Optimization Approach}

\par Various approaches have been employed to solve the complex optimization problems, such as heuristic algorithm, game theory, and deep reinforcement learning (DRL). For example, Mei \emph{et al.}~\cite{mei2023energy} presented a joint resource allocation and task offloading optimization approach by using the heuristic algorithm. Moreover, Fan \emph{et al.}~\cite{Fan2023game} employed potential game to deal with the edge competition in vehicular MEC networks. Chu \emph{et al.}~\cite{chu2023efficient} adopted the game theory to optimize the task offloading strategies for multi-user MEC system. Besides, Fang \emph{et al.}~\cite{Fang2023DRL} designed a DRL-based resource allocation and task offloading scheme to minimize the energy consumption for MEC systems. In addition, Yang \emph{et al.}~\cite{Yang2024Coo} introduced a cooperative task offloading method by utilizing the multi-agent DRL algorithm, allowing each server to make independent offloading decisions.

\par However, heuristic algorithms typically could not guarantee optimal solutions and often require numerous iterations for dynamic and large-scale scenarios, resulting in high computational overhead, extended processing delays, and inflexibility to adapt to dynamic environments. Moreover, game theory-based approaches can be mathematically complex and may struggle with scalability, particularly in scenarios with a high number of users and servers, where interactions become increasingly intricate. Besides, although DRL is powerful for learning optimal policies, it requires a substantial amount of sample data to obtain the optimal decisions, which can be problematic for real-time scenarios with limited data availability.  Additionally, DRL-based methods often lack the theoretical performance guarantees, making them less reliable in ensuring long-term system stability for dynamic scenarios. In contrast, the Lyapunov-based optimization framework is able to make real-time decisions and provide stable performance guarantees. Therefore, we propose an online joint optimization approach that leverages the Lyapunov optimization framework to ensure long-term UD energy efficiency and system stability in the dynamic and resource-constrained MEC system.

\section{System Model}
\label{sec_model}

\par As depicted in Fig. \ref{fig_systemModel}, we consider a three-layer collaborative MEC architecture, which consists of the UD layer with a set of energy-constrained UDs, denoted as $\mathcal{U} = \{1,\ldots,u,\ldots U \}$, the MEC layer with a set of MEC servers, denoted as $\mathcal{M} = \{1,\ldots,m,\ldots, M \}$, and the cloud layer with a cloud server $c$. Specifically, \textit{at the UD layer}, a set of energy-constrained devices such as wearable devices, handheld devices, and AR/VR equipment periodically generate tasks. \textit{At the  MEC layer}, base stations (BSs) are equipped with MEC servers, providing both radio access and computation services to execute the tasks generated by energy-limited UDs\footnote{We use the terms MEC server and BS interchangeably.}. \textit{At the cloud layer}, the cloud server offers additional computational resources to assist the MEC servers for processing the tasks that cannot be handled at the MEC layer. Moreover, we consider that the system time is divided into $T$ time slots $\mathcal{T}= \{1,\ldots,t,\ldots, T \}$ with equal slot duration $\delta$. In each time slot $t$, each UD $u\in \mathcal{U}$ generates a task~\cite{dai2022task} characterized by $\kappa_u(t)=(s_{u}(t), \rho)$, where $s_{u}(t)$ represents the task data size, and $\rho$ denotes the computation intensity (in cycles/bit).

%
%

\subsection{Mobility Model}
\label{sec_mobility_model}

\par Considering the temporal-dependent mobility model, we adopt the Gauss-Markovmodel~\cite{liang1999predictive} to capture the mobility of UDs. In specific, the velocity of each UD follows~\cite{Yang2022}:
\begin{sequation}
\begin{aligned}
	\label{eq_MD_mobility}
		& \mathbf{v}_{u}(t+1)=\omega\mathbf{v}_{u}(t)+(1-\omega) \bar{\mathbf{v}}_u+\sqrt{1-\omega^2} \mathbf{w}_{u}(t), \\& \forall u\in\mathcal{U}, \, t\in\mathcal{T},
\end{aligned}
\end{sequation}

\noindent where $\mathbf{v}_u(t)$ represents the velocity of UD $u$ at time slot $t$. $\omega\in[0,1]$ denotes the memory degree, which indicates the temporal-dependent level of the velocity. Additionally, $\mathbf{w}_u(t)$ means the uncorrelated random Gaussian process, i.e., $ \mathbf{w}_u(t)\sim f^{\text{Gua}}\left(0,\sigma^2\right)$, where $\sigma$ means the asymptotic standard deviation of velocity. Besides, $\bar{\mathbf{v}}_u$ indicates the asymptotic mean of velocity. Therefore, the position of each UD evolves as
\begin{sequation}
	\label{eq_MD_position}
	\begin{aligned}	\mathbf{q}_{u}(t+1)=\mathbf{q}_{u}(t) + \mathbf{v}_u(t)\delta, \, \forall u\in \mathcal{U}, \,  t \in \mathcal{T}.
	\end{aligned}
\end{sequation}

\subsection{Communication Model}
\label{sec_communication_model}

\par To mitigate interference among different MEC servers, we consider that the MEC servers use distinct frequency band. The orthogonal frequency-division multiple access (OFDMA) scheme is employed for each MEC server to simultaneously serve multiple UDs. Therefore, in time slot $t$, the transmission rate between UD $u$ and MEC server $m$ is given as
\begin{sequation}
\label{eq_date_rate}
r_{u}^{m}(t)=a_{u}^{m}(t) B_{m}\log _{2}\Big(1+\frac{p_{u}^{\text{tra}}(t)g_{u}^{m}(t)}{N_0}\Big), 
\end{sequation}

\noindent where $a_{u}^{m}(t) (0 \leq a_{u}^{m}(t) \leq 1)$ indicates the decision of communication resource allocation, which represents the fraction of bandwidth allocated to UD $u$ by MEC server $m$. Moreover, $B_{m}$ denotes the total bandwidth of the MEC server $m$, $p_{u}^{\text{tra}}(t)$ represents the transmit power of UD $u$, and $N_0$ means the Gaussian noise power. $g_{u}^{m}(t)$ represents the instantaneous channel power gain between UD $u$ and MEC server $m$, which is given as~\cite{3GPPTR389012020}
\begin{sequation}
g_{u}^{m}(t)=|h_{u}^{m}(t))|^2 ( L_{u}^m(t))^{-1},
\end{sequation}

\noindent where $h_{u}^{m}(t)$ and $L_{u}^m(t)$ represent the parameters of small-scale fading and large-scale fading, respectively. Specifically, the large-scale fading is given as
\begin{sequation}
\label{eq_LoS_probability}
    L_u^m(t)=\left(4\pi d_0 f_c/C\right)^2 \left(d_u^m(t)/d_0\right)^{\beta} 10^{\chi/10}, 
\end{sequation}
 
\noindent where $f_c$ means the carrier frequency, $C$ denotes the speed of light, $d_0$ represents the reference distance, $ d_u^m(t)$ indicates the distance between UD $u$ and MEC server $m$ in time slot $t$, $\beta$ represents the path loss exponent, and $\chi$ (in dB) is the shadowing due to blocking loss~\cite{3GPPTR389012020}, which follows a zero-mean Gaussian distributed random variable, i.e., $\chi\sim f^{\text{Gua}}(0,\sigma^2)$. Moreover, the Rayleigh fading~\cite{yu2020efficient} is employed to model the small-scale fading, which is given as 
\begin{sequation}
\label{eq_LoS_probability}
    \begin{aligned}
        f^{\text{Ray}}\left(h_{u}^{m}(t),\alpha\right) =h_{u}^{m}(t)e^{(-h_{u}^{m}(t))^2/(2\alpha^2)}/\alpha^2,
    \end{aligned}
\end{sequation}

\noindent where $\alpha$ denotes the Rayleigh fading coefficient, and $f^{\text{Ray}}$ represents the Rayleigh distribution.

\subsection{UD Energy Consumption Model}
\label{sec_comp_model}

\par A task of UD $u$ can be processed locally or offloaded to an MEC server, which is determined by the task offloading decision. Specifically, we define $x_{u}^{m}(t) \in\{0,1\}$ as the variable of task offloading decision, where $x_{u}^{m}(t)=0$ indicates that the task is processed locally on UD $u$ in time slot $t$, and $x_{u}^{m}(t)=1$ represents that the task is offloaded to MEC server $m$ in time slot $t$. Both local computing and task offloading incur energy consumption for the energy-constrained UDs, which are elaborated in the following subsections.

\par For local computing, the energy consumption of each UD is primarily determined by the task computation. Hence, the energy consumption of UD $u$ is given as
\begin{sequation}
E_{u}^{\text{loc}}(t)=(1-x_u^m(t))\zeta_uf_{u}^{2} s_{u}(t)\rho,
\end{sequation}

\noindent where $f_{u}$ (in cycles/s) represents the computation resources of UD $u$,  and $\zeta_k$ denotes the effective switched capacitance coefficient of UD $u$, which depends on the CPU architecture of UD $u$~\cite{burd1996processor}.

\par For edge computing, the energy consumption of each UD is mainly incurred by task uploading. Accordingly, the energy consumption of UD $u$ to upload task $\kappa_u(t)$ is calculated as
\begin{sequation}
E_{u}^{\text{off}}(t)=x_u^m(t)p_{u}^{\text{tra}}(t)\frac{s_{u}(t)}{r_{u}^{m}(t)}.
\end{sequation}

\begin{remark}
Note that there is a greater probability that the Rayleigh fading channel can cause channel gain to approach zero, resulting in an infinite energy consumption. However, the Rayleigh channel is more suitable to capture the multipath propagation characteristics in the considered scenario, where signals often experience varying interference due to reflections, diffractions, and scattering from surrounding buildings, trees, and other obstacles~\cite{Gupta2022Noma}. Moreover, other channel models may also lead to infinite energy consumption with a relatively small probability. In this case, we consider local computing as a potential task offloading decision to meet the energy constraints of UDs, which will be detailed in Section \ref{sec_ProblemTransfor}. 
\end{remark}


%
%
\subsection{MEC Server Queue Model}

\par At the MEC server, the task processing has several key characteristics. First, task arrivals from UDs are dynamic and unpredictable. Moreover, in practical MEC systems, task processing on edge servers is a continuous process, and the accumulation of multiple tasks can result in task backlogs and queuing delays due to the limited computing resources of the MEC servers.
Therefore, we adopt a queuing model to capture the dynamic nature of task arrivals and the associated queuing delays at the MEC servers~\cite{li2024computation}. Specifically, in time slot $t$, we consider that tasks from a set of UDs $\mathcal{U}_m(t)=\{1,\dots,U_m(t)\}$ and can be processed at MEC server $m$ or uploaded from the MEC server to cloud $c$ for collaborative execution. Consequently, the tasks are split into two parts, i.e., one part is processed at the MEC server $m$, and the other part is offloaded to cloud $c$. We denote the offloading decision from MEC server $m$ to cloud $c$ for the task of UD $u$ as $x_u^{m \rightarrow c}(t)\in\{0,1\}$, and the task splitting at MEC server $m$ in time slot $t$ is given as
\begin{sequation}
\left\{\begin{array}{l}
A_{m}^{\text{E}}(t)+A_{m}^{\text{C}}(t)=\sum_{u=1}^{U_{m}(t)} x_{u}^{m}(t) s_{u}(t), \\
A_{m}^{\text{C}}(t)=\sum_{u=1}^{U_m(t)} x_{u}^{m\rightarrow c}(t) s_{u}(t),
\end{array}\right.
\end{sequation}

\noindent where $A_{m}^{\text{E}}(t)$ denotes the amount of data processed at MEC server $m$, and $A_{m}^{\text{C}}(t)$ represents the amount of data that is offloaded from MEC server $m$ to cloud $c$.

\par To store the split tasks, each MEC server $m$ maintains two task queues, i.e., edge computing queue and cloud offloading queue for the tasks processed at MEC server $m$ and the tasks offloaded from MEC server $m$ to cloud $c$, receptively. Specifically, the length of the edge computing queue $Q_m^{\text{E}}(t)$ evolves as
\begin{sequation}
\label{eq_deadline_QL}
Q_{m}^{\text{E}}(t+1)=\max \big\{Q_{m}^{\text{E}}(t)-f_{m}\delta/\rho, 0\big\}+A_m^{\text{E}}(t),
\end{sequation}

\noindent where $f_{m}\delta/\rho$ denotes the amount of data leaving the task queue. Similarly, the length of the cloud offloading queue $Q_m^{\text{C}}(t)$ evolves as
\begin{sequation}
\label{eq_deadline_QO}
Q_{m}^{\text{C}}(t+1)=\max \big\{Q_{m}^{\text{C}}(t)-r_{m}^{c}\delta, 0\big\}+A_{m}^{\text{C}}(t),
\end{sequation}

\noindent where $r_{m}^{c}$ denotes the transmission rate from MEC server $m$ to cloud $c$.

\par To guarantee that tasks arrived at the MEC server can be processed within the tolerable delay, the queuing delay must be constrained. Based on Little's law~\cite{JW1967}, the queuing delay of the edge computing queue can be constrained as
\begin{sequation}
\label{eq_queu_deadline1}
\lim _{T \rightarrow \infty} \frac{1}{T} \sum_{t=1}^{T} Q_{m}^{\text{E}}(t)/\tilde{A}_{m}^{\text{E}}(t) \leq \bar{D}_{m}^{\text{E}}, \forall m\in \mathcal{M},
\end{sequation}

\noindent where $\bar{D}_{m}^{\text{E}}(t)$ represents the maximum tolerable queuing delay of the edge computing queue, and $\tilde{A}_{m}^{\text{E}}(t)=\frac{1}{t} \sum_{i=1}^{t} A_m^{\text{E}}(i)$ is the time-averaged data arrival rate of the edge computing queue. Similarly, the queuing delay of the cloud offloading queue can be constrained as
\begin{sequation}
\label{eq_queu_deadline2}
\lim _{T \rightarrow \infty} \frac{1}{T} \sum_{t=1}^{T} Q_{m}^{\text{C}}(t)/\tilde{A}_{m}^{\text{C}}(t) \leq D_{m, \max }^{\text{C}}, \forall m\in \mathcal{M},
\end{sequation}

\noindent where $\bar{D}_{m}^{\text{C}}$ denotes the maximum tolerable queuing delay of the cloud offloading queue, and $\tilde{A}_{m}^{\text{C}}(t)=\frac{1}{t} \sum_{i=1}^{t} A_{m}^{\text{C}}(i)$ is the time-averaged data arrival rate  for this queue.

\section{Problem Formulation and Analysis}
\label{sec_pro_analysis}

\par The problem formulation and problem analysis are presented in this section.

\subsection{Problem Formulation}

\par The objective of this work to minimize the average energy consumption of UDs over time by jointly optimizing the communication resource allocation $\mathbf{A}={\{\mathbf{A}(t)|\mathbf{A}(t)=\{a_u^m(t)\}_{u\in\mathcal{U},m\in\mathcal{M}}\}}_{t\in\mathcal{T}}$ and task offloading $\mathbf{X}={\{\mathbf{X}(t)|\mathbf{X}(t)=\{x_u^m(t),x_u^{m\rightarrow c}(t)\}_{u\in\mathcal{U},m\in\mathcal{M}}\}}_{t\in\mathcal{T}}$. Therefore, the optimization problem is formulated as

\vspace{-0.8em}
{\small
\begin{align} \mathbf{P}:&\min_{\mathbf{X,A}}\frac{1}{T} \sum_{t=1}^{T}\sum_{m=1}^{M} \sum_{u=1}^{U_{m}(t)} \mathbb{E}\big\{E_{u}^{\text{loc}}(t)+E_{u}^{\text{off}}(t)\big\}\label{P1}  \\
\text {s.t. }&  x_{u}^{m}(t),x_{u}^{m \rightarrow c}(t) \in\{0,1\}, \, \forall m \in \mathcal{M}, \,  u \in \mathcal{U}_m(t),\,  t \in \mathcal{T}, \tag{\ref{P1}{a}} \label{eq_constraint_x} \\
& x_{u}^{m \rightarrow c}(t) \leq x_{u}^{m}(t),\forall m \in \mathcal{M}, \, u \in \mathcal{U}_m(t),\, t \in \mathcal{T}, \tag{\ref{P1}{b}}\label{constraint_m_to_c}\\
& \sum_{u \in \mathcal{U}_{m}(t)} x_{u}^{m}(t) \leq C_m,\forall m \in \mathcal{M}, \, t \in \mathcal{T}, \tag{\ref{P1}{c}}\label{eq_constraint_x_conn}\\
&  a_{u}^{m}(t) \in \left[0,1\right], \forall m \in \mathcal{M}, \,  u \in \mathcal{U}_m(t),\, t \in \mathcal{T}, \tag{\ref{P1}{d}}\label{eq_constraint_com}\\
& a_{u}^{m}(t) \leq x_{u}^{m}(t),\forall m \in \mathcal{M}, \, u \in \mathcal{U}_m(t),\, t \in \mathcal{T}, \tag{\ref{P1}{e}}\label{eq_constraint_com1}\\
& \sum_{u \in \mathcal{U}_{m}(t)} x_{u}^{m}(t) a_{u}^{m}(t) \leq 1, \forall m \in \mathcal{M}, \, t \in \mathcal{T}, \tag{\ref{P1}{f}}\label{eq_constraint_com_and_off}\\
&\eqref{eq_MD_mobility}, \eqref{eq_MD_position}, \eqref{eq_queu_deadline1} \text{ and } \eqref{eq_queu_deadline2}, \notag
\end{align}}

\noindent where constraint \eqref{eq_constraint_x} indicates that the decision of task offloading is binary and \eqref{constraint_m_to_c} is the constraint on task offloading decisions at the MEC server side. Moreover,  \eqref{eq_constraint_x_conn} limits the maximum number of connections of MEC server $m$, where $C_m$ represents the connection capacity of MEC server $m$. Furthermore, \eqref{eq_constraint_com} to \eqref{eq_constraint_com_and_off} constrain the communication resource allocation. Additionally, \eqref{eq_MD_mobility} and \eqref{eq_MD_position} represent the mobility of UDs. Besides, \eqref{eq_queu_deadline1} and \eqref{eq_queu_deadline2} limit the queuing delay at the MEC server side. 

\subsection{Problem Analysis}

\par We analyze the issues of solving the formulated problem and the motivation of the proposed approach.

\subsubsection{Challenges} 

\par Solving problem $\mathbf{P}$ directly presents several challenges.

\begin{itemize}

\item \textit{Future-dependent nature in a dynamic environment.} Solving problem $\mathbf{P}$ requires complete knowledge of future information, such as the offloading demands of UDs and the workload of MEC servers. However, in our considered MEC system, the mobility of UDs, wireless channel conditions, task arrivals, and available MEC server resources fluctuate unpredictably over time, making it difficult to predict future states accurately. Therefore, without the knowledge of future conditions, it is  complicated to make optimal decisions in advance.

\item \textit{MINLP problem and coupled decision variables.} First, problem $\mathbf{P}$ involves both binary decision variables (i.e., task offloading $\mathbf{X}$) and continuous decision variables (i.e., communication resource allocation $\mathbf{A}$). This classifies problem $\mathbf{P}$ an NP-hard and non-convex MINLP problem~\cite{Boyd2014} due to the combinatorial complexity introduced by the binary variables and the non-linear relationships associated with the continuous variables. Furthermore, the coupling of decision variables significantly complicate the problem optimization process.


\end{itemize}

\subsubsection{Motivation} 

\par The challenges mentioned above necessitate the design of an efficient approach to solve problem $\mathbf{P}$, driven by the following key motivations.

\par \textit{i) Overcoming future-dependence with online decision-making.} The future-dependent nature of problem $\mathbf{P}$ necessitates an online approach that can make real-time decisions without complete knowledge of future system information. Although machine learning algorithms such as DRL have shown effectiveness in decision-making, the mixed integral and coupled decision variables of problem $\mathbf{P}$ could lead to complex action spaces. As a result, DRL approaches require extensive environmental interactions and long convergence times to learn optimal policies, which is costly in the resource-constrained MEC system. Moreover, DRL struggles with the dynamic state space caused by user mobility.  In comparison, the Lyapunov optimization framework offers several advantages in addressing these challenges. 

\begin{itemize}[leftmargin=2em]
\item No need for future information. The Lyapunov optimization framework transforms a long-term optimization problem into a sequence of real-time subproblems, enabling decisions without future knowledge\cite{liu2019dynamic}. This makes it well-suited for the dynamic MEC environment where future conditions are uncertain.

\item Low complexity and stability. Compared to DRL, the Lyapunov optimization framework has lower computational overhead by continuously adjusting the decisions, as it continuously adjusts decisions based on current observations, making it more appropriate for real-time decision-making in dynamic MEC systems \cite{guo2024lyapunov}. Additionally, the Lyapunov optimization framework guarantees stability, ensuring convergence and boundedness of system dynamics. 

\item Handling long-term queuing constraints. The Lyapunov optimization framework can effectively manage the long-term queuing constraints at the MEC server by ensuring that a decision of each time slot leads to an optimal or near-optimal solution in the long term.
\end{itemize}

\textit{ii) Decoupling interdependent decision variables.} The interdependency of communication resource allocation and task offloading motivates the decoupling of problem $\mathbf{P}$ into  simpler subproblems. Specifically, we use the Tammer decomposition mechanism to decouple the original problem into a task offloading subproblem with binary decision variables and a communication resource allocation subproblem with continuous variables. The Tammer decomposition mechanism brings several advantages in the problem decoupling~\cite{Tammer1987}.

\begin{itemize}[leftmargin=2em]
\item Solving of subproblems in parallel. The resulting subproblems of the Tammer decomposition mechanism can often be solved in parallel. This parallelization is important in the considered resource-constrained and dynamic MEC system, where real-time decision-making is essential. 

\item Optimal or near-optimal solutions. While the Tammer decomposition mechanism simplifies the problem, it ensures that the resulting subproblems still achieve optimal or near-optimal solutions for the original problem \cite{Tammer1987}. As a result, solving the subproblems of task offloading and resource allocation independently simplifies the decision-making process, enabling more efficient optimization while satisfying the task offloading requirements of UDs and  the resource constraints of MEC servers.
\end{itemize}

%
%
\section{The Proposed OJCTA}
\label{sec_ProblemTransfor}

\par In this section, we propose OJCTA to solve the EEDAOP,  and the framework of OJCTA is shown in Fig. \ref{fig_framwork}. Specifically, we first transform problem $\mathbf{P}$ into a per-time-slot real-time optimization problem by leveraging the Lyapunov optimization framework. Then, we propose a communication resource allocation and task offloading for the real-time optimization problem.

\begin{figure*}[h]
    \centering
    \setlength{\abovecaptionskip}{0pt}%
    \setlength{\belowcaptionskip}{2pt}%
    \includegraphics[width =7in]{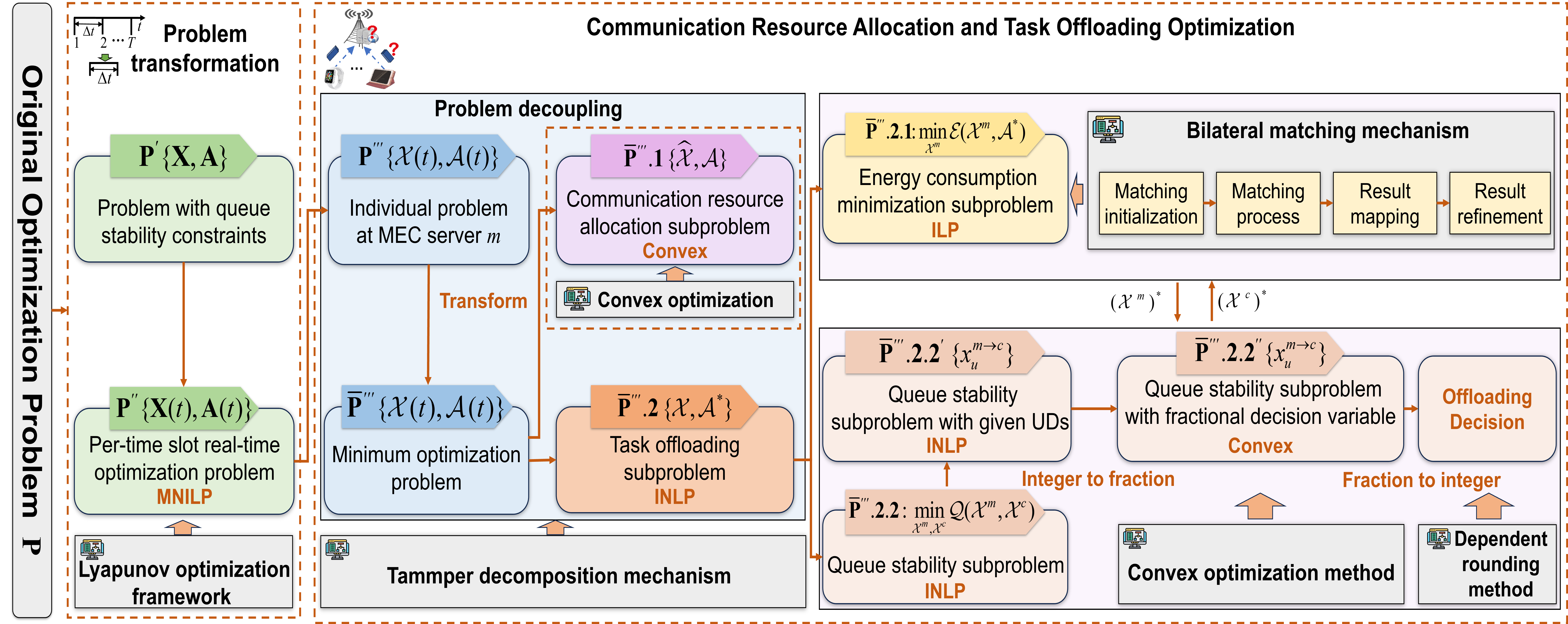}
    \caption{The framework of OJCTA. The EEDAOP is first transformed to a per-time-slot real-time optimization problem, and then solved by a communication resource allocation and task offloading optimization method.}
\label{fig_framwork}
\end{figure*}
\vspace{-1em}

%
%
\subsection{Problem Transformation}

\par To deal with the future-dependence of problem $\mathbf{P}$, we decouple $\mathbf{P}$ into a per-time-slot real-time optimization problem using the Lyapunov optimization framework. 


\par \textbf{First}, we transform  the long-term constraints (\ref{eq_queu_deadline1}) and (\ref{eq_queu_deadline2}) into queue stability constraints by using the virtual queue. Specifically, we introduce the virtual queues $Z_{m}^{\text{E}}(t)$ and $Z_{m}^{\text{C}}(t)$ for (\ref{eq_queu_deadline1}) and (\ref{eq_queu_deadline2}),  respectively. Then the dynamic of virtual queues is updated as

\vspace{-0.8em}
{\small
\begin{subequations}
\label{eq_virtual}
\begin{alignat}{1}
        &Z_{m}^{\text{E}}(t+1)=\max \Big\{Z_{m}^{\text{E}}(t)+\frac{Q_{m}^{\text{E}}(t)}{\tilde{A}_{m}^{\text{E}}(t)}-\bar{D}_{m}^{\text{E}}, 0\Big\}, \forall m \in \mathcal{M}, t \in \mathcal{T},\label{eq_virtual_L}\\
       & Z_{m}^{\text{C}}(t+1)=\max \Big\{Z_{m}^{\text{C}}(t)+\frac{Q_{m}^{\text{C}}(t)}{\tilde{A}_{m}^{\text{C}}(t)}-\bar{D}_{m}^{\text{C}}, 0\Big\}, \forall m \in \mathcal{M}, t \in \mathcal{T}\label{eq_virtual_O}.
\end{alignat}
\end{subequations}
}

\par  Based on Eqs.~\eqref{eq_virtual_L} and \eqref{eq_virtual_O}, constraints (\ref{eq_queu_deadline1}) and (\ref{eq_queu_deadline2})  can be converted into the stability constraints for queues $Z_{m}^{\text{E}}$ and $Z_{m}^{\text{C}}$, receptively. Therefore, problem $\mathbf{P}$ is transformed as

\vspace{-0.8em}
{\small
\begin{align} \mathbf{P^{\prime}}:&\min_{\mathbf{X,A}}\frac{1}{T} \sum_{t=1}^{T}\sum_{m=1}^{M} \sum_{u=1}^{U_{m}(t)} \mathbb{E}\big\{E_{u}^{\text{loc}}(t)+E_{u}^{\text{off}}(t)\big\}\label{P2}  \\
\text {s.t.}\, &  \eqref{eq_MD_mobility}, \, \eqref{eq_MD_position}, \, \eqref{eq_constraint_x} \sim \eqref{eq_constraint_com_and_off}, \, \eqref{eq_virtual_L}, \text{ and } \eqref{eq_virtual_O}.\notag
\end{align}} 

\par \textbf{Second}, to achieve a scalar measure of the queue backlogs, we define the \textit{Lyapunov function} $L(\Theta(t))$ as
\begin{sequation}
\begin{aligned}
& L(\Theta(t))=\frac{1}{2} \sum_{m=1}^{M}\big((Q_m^{\text{E}}(t))^2+(Q_m^{\text{C}}(t))^2+(Z_{m}^{\text{E}}(t))^2+(Z_{m}^{\text{C}}(t))^2\big),
\end{aligned}
\end{sequation}

\noindent where $\Theta(t)=\{Q_{m}^{\text{E}}(t), Q_{m}^{\text{C}}(t), Z_{m}^{\text{E}}(t), Z_{m}^{\text{C}}(t)\}$ represents the vector of current queue backlogs.

\par \textbf{Third}, ensuring task queue stability is essential to prevent excessive queuing delays and task backlogs. Therefore, to maintain the stability of the task queues (i.e., $Q_m^{\text{E}}(t)$ and $Q_m^{\text{C}}(t)$) and virtual queues (i.e., $Z_m^{\text{O}}(t)$ and $Z_m^{\text{O}}(t)$), the \textit{conditional Lyapunov drift} is defined as~\cite{2010Neely}:
\begin{sequation}
\Delta(\Theta(t)) = \mathbb{E}[L(\Theta(t+1))-L(\Theta(t)) \mid \Theta(t)],
\end{sequation}

\noindent which captures the expected change of the Lyapunov function across consecutive time slots.

\par \textbf{Fourth}, to minimize the average energy consumption of UDs under the queue stability constraints, we define the \textit{Lyapunov drift-plus-penalty function} as follows~\cite{2010Neely}:
\begin{sequation}
\begin{aligned}
& \Lambda(\Theta(t)) =\Delta(\Theta(t))+V\sum_{m=1}^{M}\sum_{u=1}^{{U}_{m}(t)} \mathbb{E}\big\{E_{u}^{\text{loc}}(t)+E_{u}^{\text{off}}(t) \mid \Theta(t)\big\},
\end{aligned}
\end{sequation}


\noindent where $V\geq0$ is a nonnegative control parameter to balance the UD energy consumption and the queue stability. Next, we derive an upper bound of $\Lambda(\Theta(t)) $ in Theorem \ref{the_bound_drift}.

\begin{theorem}
\label{the_bound_drift}
For all all possible queue backlogs $\Theta(t)$, $t$, and $V\geq0$, the Lyapunov drift-plus-penalty function is upper bounded by 
\begin{sequation}
\label{eq_drift_bound}
\begin{aligned}
\Lambda(\Theta(t))&\leq B+\sum_{m=1}^{M}\Big(V\sum_{u=1}^{U_m(t)} \mathbb{E}\big\{E_{u}^{\text{loc}}(t)+E_{k}^{\text{off}}(t) \mid \Theta(t)\big\} \\&+Q_{m}^{\text{E}}(t) \mathbb{E}\big\{A_m^{\text{E}}(t)--\frac{f_m\delta}{\rho} \mid \Theta(t)\big\}\\
&+Q_{m}^{\text{C}}(t) \mathbb{E}\big\{A_{m}^{\text{C}}(t)-r_{m}^c\delta \mid \Theta(t)\big\}\\
&+Z_{m}^{E}(t) \mathbb{E}\Big\{\frac{Q_{m}^{\text{E}}(t)}{\tilde{A}_{m}^{L}(t)}-\bar{D}_{m}^{\text{E}} \mid \Theta(t)\Big\}\\
&+Z_{m}^{\text{C}}(t)\mathbb{E}\Big\{\frac{Q_{m}^{\text{C}}(t)}{\tilde{A}_{m}^{L}(t)}-\bar{D}_{m}^{\text{C}} \mid \Theta(t)\Big\}\Big),
\end{aligned}
\end{sequation}

\noindent where $B$ is a finite positive constant which is lower bounded by:
\begin{sequation}
\begin{aligned}
&B \geq \frac{1}{2}\sum_{m=1}^{M}\Big(\mathbb{E}
\big\{[A_m^{\text{E}}(t)]^{2}+[\frac{f_m\delta}{\rho}]^{2} \mid \Theta(t)\big\} \\
&  +\mathbb{E}\big\{\big[A_{m}^{\text{C}}(t)\big]^{2}+[r_{m}^c\delta]^{2} \mid \Theta(t)\big\} \\
& +\mathbb{E}\big\{\big[\frac{Q_{m}^{\text{E}}(t)}{\tilde{A}_{m}^{E}(t)}\big]^{2}+[\bar{D}_{m}^{\text{E}}]^{2} \mid \Theta(t)\big\} \\
& +\mathbb{E}\big\{\big[\frac{Q_{m}^{\text{C}}(t)}{\tilde{A}_{m}^{\text{C}}(t)}\big]^{2}+[\bar{D}_{m}^{\text{C}}]^{2} \mid \Theta(t)\big\}\Big).
\end{aligned}
\end{sequation}
\end{theorem}

\begin{proof}
 The details can be found in Appendix \ref{app_the_bound_drift} of the supplemental material. 
\end{proof}
\vspace{-0.5em}

\par \textbf{Finally}, according to the Lyapunov optimization framework, problem $\mathbf{P^{\prime}}$ can be solved by minimizing the right-hand side of \eqref{eq_drift_bound}. Therefore, problem $\mathbf{P^{\prime}}$, which depends on the future information, is transformed into a real-time optimization problem that can be solved by using current information, as follows:

\vspace{-0.8em}
{\small
\begin{align}
\label{eq_pro_Lya}
&\mathbf{P^{\prime\prime}}: \min_{\mathbf{X}(t), \mathbf{A}(t)} \sum_{m=1}^M\bigg\{V  \sum_{ k = 1 } ^ {U_m(t)} \Big[p_u^\text{tra}  \frac{x_u^m(t) s_u(t)}{r_u^m(t)}\nonumber  \\
& +\zeta_uf_u^3  \frac{\left(1-x_u^m(t)\right) s_u(t)\rho}{f_u}\Big] \nonumber\\
& +Q_m^{\text{E}}(t)\big[\sum_{u=1}^{U_m(t)}(x_u^m(t)-x_u^{m \rightarrow c}(t)) s_u(t)-\frac{f_m\delta}{\rho}  \big] \nonumber\\
& +Q_m^{\text{C}}(t)\big[\sum_{u=1}^{U_m(t)} x_u^{m \rightarrow c}(t) s_u(t)-r_m^c  \delta\big] \nonumber\\
& +Z_{m}^{\text{E}}(t) \frac{Q_m^{\text{E}}(t)}{\frac{1}{t}\big[\sum\limits_{i=1}^{t-1} A_m^{\text{E}}(i)+\sum\limits_{u=1}^{U_m(t)}(x_u^m(t)-x_u^{m \rightarrow c}(t)) s_u(t)\big]} \nonumber\\
& +Z_{m}^{\text{C}}(t) \frac{Q_m^{\text{C}}(t)}{\frac{1}{t}\big[\sum\limits_{i=1}^{t-1} A_m^{\text{C}}(i)+\sum\limits_{u=1}^{U_m(t)}x_u^{m \rightarrow c}(t)s_u(t)\big]}\bigg\} \\
& \text{s.t.} \, \eqref{eq_MD_mobility}, \, \eqref{eq_MD_position},  \,  \eqref{eq_constraint_x} \sim \eqref{eq_constraint_com_and_off} \, \notag.
\end{align}}

\par However, problem $\mathbf{P}^{\prime\prime}$ is still an MINLP problem, with coupled decision variables of task offloading $x(t)$ and communication resource allocation $a(t)$ at time slot $t$. Consequently, the mutual dependence of decision variables could complicate the process of finding the optimal solution of $\mathbf{P}^{\prime\prime}$. To address this challenge, we propose a  communication resource allocation and task offloading optimization method that decouples the decision variables, allowing us to achieve a sub-optimal solution with polynomial time complexity. 

%
%
 
\subsection{Communication Resource Allocation and Task Offloading Optimization Method}

\par This section presents a communication resource allocation and task offloading optimization method is presented. Specifically, to deal with the decoupled decision variables in problem $\mathbf{P}^{\prime\prime}$, we first decouple problem $\mathbf{P}^{\prime\prime}$ into the subproblems of communication resource allocation and task offloading by adopting the Tammer decomposition mechanism. Then, we present a convex optimization-based method for the subproblem of communication resource allocation. Finally, we propose a two-stage alternating optimization method for the subproblem of task offloading.

\subsubsection {Problem Decoupling}
\label{sec_Problem Decoupling}

\par  We know from \eqref{eq_pro_Lya} that problem $\mathbf{P}^{\prime\prime}$ involves the binary decision variable of task offloading $\mathbf{X}$ and continuous decision variable of communication resource allocation $\mathbf{A}$. Therefore, we employ the Tammer decomposition mechanism to decouple problem $\mathbf{P}^{\prime\prime}$ into a communication resource allocation subproblem and a task offloading subproblem, which is detailed as follows.


\par First, by considering that each MEC server operates independently with its own resources and workloads, problem $\mathbf{P}^{\prime\prime}$ can be decomposed into parallel subproblems for each individual MEC server. Specifically, the individual problem at MEC server $m$ is given as

\vspace{-0.8em}
{\small
\begin{align}
&\mathbf{P}^{\prime\prime\prime}: \min_{\mathbf{X}(t), \mathbf{A}(t)} V \sum_{u=1}^{U_m(t)}\Big[x_u^m(t) p_u^\text{tra}(t)  \frac{s_u(t)}{r_u^m(t)}\notag\\
&+(1-x_u^m(t)) \zeta_uf_u^2 s_u(t) \rho\Big] \nonumber\\
& +Q_m^{\text{E}}(t)\big[\sum_{u=1}^{U_m(t)}(x_u^m(t)-x_u^{m \rightarrow c}(t)) s_u(t)-\frac{f_m \delta}{\rho} \big] \nonumber\\
& +Q_m^{\text{C}}(t)\big[\sum_{u=1}^{U_m(t)} x_u^{m \rightarrow c}(t) s_u(t)-r_m^c  \delta\big] \nonumber\\
& +Z_{m}^{\text{E}}(t)\frac{Q_m^{\text{E}}(t)}{\frac{1}{t}\big[\sum\limits_{i=1}^{t-1} A_m^{\text{E}}(i)+\sum\limits_{u=1}^{U_m(t)}(x_u^m(t)-x_u^{m \rightarrow c}(t)) s_u(t)\big]} \nonumber\\
& +Z_{m}^{\text{C}}(t)\frac{Q_m^{\text{C}}(t)}{\frac{1}{t}\big[\sum\limits_{i=1}^{t-1} A_m^{\text{C}}(i)+\sum\limits_{u=1}^{U_m(t)} x_u^{m \rightarrow c}(t) s_u(t)\big]} \label{eq_individual_pro}\\
& \text { s.t. }  \eqref{eq_MD_mobility}, \, \eqref{eq_MD_position}, \, \eqref{eq_constraint_x} \sim \eqref{eq_constraint_com_and_off} \notag.
\end{align}}

\par  Then, for given communication resource allocation decision $\mathcal{A}$ and task offloading decision $\mathcal{X}$, we denote the objective function of problem $\mathbf{P}^{\prime\prime\prime}$ as $J(\mathcal{X}, \mathcal{A})$. Therefore, problem $\mathbf{P}^{\prime\prime\prime}$ is transformed into a minimum optimization problem as:
\begin{sequation}
\begin{aligned}
\mathbf{\bar{P}^{\prime\prime\prime}}:& \min _{\mathcal{X}}\left(\min _{\mathcal{A}} J(\mathcal{X}, \mathcal{A})\right) \\
& \text{s.t. } \eqref{eq_MD_mobility}, \, \eqref{eq_MD_position}, \, \eqref{eq_constraint_x} \sim \eqref{eq_constraint_com_and_off} \notag.
\end{aligned}
\end{sequation}

\par Moreover, given any task offloading decision $\hat{\mathcal{X}}$, we define the following partitions of UDs. We denote the set of UDs that offload tasks to MEC server $m$ as $\mathcal U_m^{\text{O}}=\left\{u \in \mathcal{U}_m(t) \mid x_{u}^{m}(t)=1\right\}$, and denote the UDs that execute tasks locally as $\mathcal U_m^L=\left\{u \in \mathcal{U}_m(t) \mid x_{u}^{m}(t) =0\right\}$ ($\mathcal U_m^{\text{O}} \cup \mathcal U_m^L = \mathcal{U}_m(t)$). We denote the set $\mathcal U_m^{\text{O}}(t)$ is further divided into two subsets, i.e., $\mathcal U_m^S=\left\{u \in \mathcal{U}_m(t) \mid x_{u}^{m}(t)=1 \land {x_u^{m \rightarrow c}=0}\right\}$ consists of UDs whose tasks are executed directly by MEC server $m$, and $\mathcal U_m^C=\left\{u \in \mathcal{U}_m(t) \mid x_{k}^{m}(t)=1 \land {x_{k}^{m \rightarrow c}(t)=1}\right\}$ consists of UDs whose tasks are offloaded from MEC server $m$ to the cloud ($\mathcal U_m^S \cup \mathcal U_m^C = \mathcal U_m^{\text{O}}$). Therefore, by removing the constraints that are irrelevant with the communication resource allocation decision, the subproblem of communication resource allocation is given as

\vspace{-0.8em}
{\small
\begin{align}
\mathbf{\bar{P}^{\prime\prime\prime}.1}:& \min_{\mathcal{A}} J(\hat{\mathcal{X}}, \mathcal{A})\label{RA}  \\
\text {s.t. }& a_{u}^{m}(t)> 0,\forall u \in \mathcal U_m^{\text{O}},\forall t \in \mathcal{T} \tag{\ref{RA}{a}}\label{RA_a},\\
& \sum_{u \in \mathcal U_m^{\text{O}}} a_{u}^{m}(t) \leq 1, \forall t \in \mathcal{T} \tag{\ref{RA}{b}} \label{RA_b}.
\end{align}}

\par Denote the optimal communication resource allocation $\mathcal{A}^*$ for problem $\mathbf{\bar{P}^{\prime\prime\prime}.1}$. Based on the solution $\mathcal{A}^*$, and eliminating the constraints , the subproblem of task offloading is represented as

\vspace{-0.8em}
{\small
\begin{align}
&\mathbf{\bar{P}^{\prime\prime\prime}.2}:\min_{\mathcal{X}} J(\mathcal{X}, \mathcal{A}^*)\label{TO}  \\
\text {s.t. }& x_{u}^{m}(t)\in\{0,1\},\forall u \in \mathcal{U}_m(t),\forall t \in \mathcal{T} \tag{\ref{TO}{a}}\label{TOa} \\
&x_{u}^{m \rightarrow c}(t) \in\{0,1\},\forall u \in \mathcal{U}_m(t),\forall t \in \mathcal{T} \tag{\ref{TO}{b}}\label{TOb} \\
& x_{u}^{m \rightarrow c}(t) \leq x_{u}^{m}(t),\forall u \in \mathcal{U}_m(t),\forall t \in \mathcal{T} \tag{\ref{TO}{c}}\label{TOc}.
\end{align}}

\par Note that decomposing problem $\mathbf{P}^{\prime\prime\prime}$ into subproblems $\mathbf{\bar{P}^{\prime\prime\prime}.1}$ and $\mathbf{\bar{P}^{\prime\prime\prime}.2}$ preserves the optimality of the solution~\cite{Tammer1987}. This is because the decision variables are considered together throughout the decoupling. More specific, the optimization of one decision variable is performed based on the optimal outcome of the other decision variable.

\subsubsection {Communication Resource Allocation}
\label{sec_resource allocation}

\par In this subsection, we optimize the decision of communication resource allocation by solving subproblem $\mathbf{\bar{P}^{\prime\prime\prime}.1}$. In specific, by substituting \eqref{eq_individual_pro} into \eqref{RA}, the subproblem of communication resource allocation is rewritten as 

\vspace{-0.8em}
{\small
\begin{align}
\mathbf{\bar{P}^{\prime\prime\prime}.1}:\, &\min_{\mathcal{A}}J(\hat{\mathcal{X}}, \mathcal{A})=\min_{\mathcal{A}}\sum_{u \in \mathcal U_m^{\text{O}}}\frac{ V \beta _u(t) }{a_{u}^{m}(t)B_m}+\phi, \label{RA1}\\
\text {s.t. }& \eqref{RA_a} \text{ and } \eqref{RA_b} \notag,
\end{align}}

\noindent where $\phi$ is a term independent of $\mathcal{A}$, {which is given as}

\vspace{-0.8em}
{\small
\begin{align}
\phi &= \sum\nolimits_{u\in \mathcal U_m^{\text{L}}(t)}V\left(1-x_u^m(t)\right) \zeta_uf_u^2 s_u(t) \rho \nonumber\\
& +Q_m^{\text{E}}(t)\left(\sum\nolimits_{u\in \mathcal U_m^{\text{S}}(t)} s_u(t)-\frac{f_m \delta}{\rho} \right) \nonumber\\
& +Q_m^{\text{C}}(t)\left(\sum\nolimits_{u\in \mathcal U_m^{\text{C}}(t)} s_u(t)-r_m^c  \delta\right) \nonumber\\
& +Z_{m}^{\text{E}}(t)\frac{Q_m^{\text{E}}(t)}{\frac{1}{t}\left[\sum_{i=1}^{t-1} A_m^{\text{E}}(i)+\sum _{u\in \mathcal U_m^{\text{S}}(t)} s_u(t)\right]} \nonumber\\
& +Z_{m}^{\text{C}}(t)\frac{Q_m^{\text{C}}(t)}{\frac{1}{t}\left[\sum_{i=1}^{t-1} A_m^{\text{C}}(i)+\sum_{u\in \mathcal{U}_m^\text{C}}(t)s_u(t)\right]} \\
& \text { s.t. }  \eqref{eq_MD_mobility}, \, \eqref{eq_MD_position}, \, \eqref{eq_constraint_x} \sim \eqref{eq_constraint_com_and_off} \notag.
\end{align}}

\noindent Additionally, $\beta_u(t)$ depends solely on the parameters of UD $u$, which is given as
\begin{sequation}
\beta_u(t) = \frac{p_u^{\text{tra}}(t)s_u(t)}{\log _{2}\big(1+\frac{p_{u}^{\text{tra}}(t)g_{u}^{m}(t)}{N_0}\big)}.
\end{sequation}

\par $\mathbf{\bar{P}^{\prime\prime\prime}.1}$ is a convex optimization problem, which is proved in Theorem \ref{lem_P1_convex}. Therefore, we solve problem $\mathbf{\bar{P}^{\prime\prime\prime}.1}$ by employing the convex optimization method~\cite{Boyd2014}, and the optimal communication resource allocation is presented in Theorem \ref{theo_opt_resource}. 

\begin{theorem}
\label{lem_P1_convex}
Problem $\mathbf{\bar{P}^{\prime\prime\prime}.1}$ is a convex optimization problem.
\end{theorem}

\begin{proof}
The details can be found in Appendix \ref{app_lem_P1_convex} of the supplemental material. 
\end{proof}
\vspace{-0.5em}

\begin{theorem}
\label{theo_opt_resource}
    The optimal communication resource allocation for problem $\mathbf{\bar{P}^{\prime\prime\prime}.1}$ is given as
    \begin{sequation}
    \label{eq_theo_opt_resource}
    a_{u}^{m}{(t)}^*= \left\{\begin{array}{l}
    \frac{\sqrt{\beta_{u}(t)} }{\sum_{u \in \mathcal U_m^{\text{O}}} \sqrt{\beta_{u}(t)}},\, \text{If } u\in \mathcal U_m^{\text{O}}, \\
    0, \, \text{If } u\notin \mathcal U_m^{\text{O}}.
    \end{array}\right.
    \end{sequation}
\end{theorem}

\begin{proof}
The details can be found in Appendix \ref{app_theo_opt_resource} of the supplemental material. 
\end{proof}
\vspace{-0.5em}

\subsubsection{Task Offloading}

\par We optimize the task offloading decisions from UD to MEC server and from MEC server to cloud by proposing a two-stage alternating optimization method. Specifically, we divide the subproblem $\mathbf{\bar{P}^{\prime\prime\prime}.2}$ into an energy consumption minimization problem and a queue stability subproblem. In the first stage, we optimize energy consumption using a many-to-one matching method. In the second stage, we address the queue stability subproblem through the convex optimization method. By iteratively applying these two methods, we achieve an optimized solution for the task offloading subproblem. We elaborate the main process of task offloading optimization as follows. 

\par By combing \eqref{eq_individual_pro} with \eqref{TO}, the subproblem of task offloading is given as 

\vspace{-0.8em}
{\small
\begin{align}
 \label{TO1}
&\mathbf{\bar{P}^{\prime\prime\prime}.2}:\, \min_{\mathcal{X}} J(\mathcal{X}, \mathcal{A}^*)=\min_{\mathcal{X}}V \sum_{u=1}^{U_m(t)}\Big[x_u^m(t) p_u^\text{tra}(t)  \frac{s_u(t)}{{r_u^m(t)}^*}\notag\\
&+ (1-x_u^m(t)) \zeta_uf_u^2 s_u(t) \rho\Big] \nonumber\\
& +Q_m^{\text{E}}(t)\big[\sum_{u=1}^{U_m(t)}(x_u^m(t)-x_u^{m \rightarrow c}(t)) s_u(t)-\frac{f_m \delta}{\rho} \big] \nonumber\\
& +Q_m^{\text{C}}(t)\big[\sum_{u=1}^{U_m(t)} x_u^{m \rightarrow c}(t) s_u(t)-r_m^c  \delta\big] \nonumber\\
& +Z_{m}^{\text{E}}(t)\frac{Q_m^{\text{E}}(t)}{\frac{1}{t}\big [\sum\limits_{i=1}^{t-1} A_m^{\text{E}}(i)+\sum\limits_{u=1}^{U_m(t)}(x_u^m(t)-x_u^{m \rightarrow c}(t)) s_u(t)\big]} \nonumber\\
& +Z_{m}^{\text{C}}(t)\frac{Q_m^{\text{C}}(t)}{\frac{1}{t}\big[\sum\limits_{i=1}^{t-1} A_m^{\text{C}}(i)+\sum\limits_{u=1}^{U_m(t)} x_u^{m \rightarrow c}(t) s_u(t)\big]} \\
&\ \ \text {s.t. }
 \eqref{TOa}\text{ and } \eqref{TOb} \notag,
\end{align}}

\noindent where ${r_u^m(t)}^*$ is the outcome of transmission rate with the optimal communication resource allocation $a_{u}^{m}{(t)}^*$. However, it is challenging to solve subproblem $\mathbf{\bar{P}^{\prime\prime\prime}.2}$ directly since it is an INLP problem, as proved in Theorem \ref{lem_P2_INLP}.

\begin{theorem}
\label{lem_P2_INLP}
Problem $\mathbf{\bar{P}^{\prime\prime\prime}.2}$ is an integer nonlinear programming (INLP) problem.
\end{theorem}

\begin{proof}
The details can be found in Appendix \ref{app_lem_P2_INLP} of the supplemental material. 
\end{proof}
\vspace{-0.5em}

\par From \eqref{TO1}, the objective of problem $\mathbf{\bar{P}^{\prime\prime\prime}.2}$ is to minimize UD energy consumption and Lyapunov drift function simultaneously. Therefore, we further decompose the task offloading problem into two subproblems, i.e., the energy consumption minimization problem for UDs and the queue stability subproblem for MEC servers, which is detailed as follows.

\par We first divide the task offloading decision as $\mathcal{X}=\mathcal{X}^m\cup\mathcal{X}^c$, where $\mathcal{X}^m=\{x_u^m(t)|\forall u\in \mathcal{U}_m(t)\}$ and $\mathcal{X}^c=\{x_u^{m \rightarrow c}(t)|\forall u\in \mathcal{U}_m(t)\}$ represent the offloading decisions at the UD side and MEC server side, respectively. More specific, $\mathcal{X}^m$ denotes the task offloading decision of each UD to  MEC server $m$, and $\mathcal{X}^c$ represents the task offloading decision of each MEC server $m$ to the cloud. Then, the objective function of problem $\mathbf{\bar{P}^{\prime\prime\prime}.2}$ is re-expressed as

\vspace{-0.8em}
{\small
\begin{align}
\mathbf{\bar{P}^{\prime\prime\prime}.2}:\,& \min_{\mathcal{X}} J(\mathcal{X}, \mathcal{A}^*)=V\mathcal{E} (\mathcal{X}^m, \mathcal{A}^*)+\mathcal{Q}(\mathcal{X}^m, \mathcal{X}^c)\\
&\text {s.t. }
 \eqref{TOa}\sim \eqref{TOc}, \notag
 \label{TO11}
\end{align}
}

\noindent where $\mathcal{E} (\mathcal{X}^m, \mathcal{A}^*)$ is given as
\begin{sequation}
\begin{aligned}
&\mathcal{E} (\mathcal{X}^m, \mathcal{A}^*)=\min_{\mathcal{X}}V \sum_{u=1}^{U_m(t)}\big[x_u^m(t) p_u^\text{tra}(t)  \frac{s_u(t)}{{r_u^m(t)}^*}\\
&+ \left(1-x_u^m(t)\right) \zeta_uf_u^2 s_u(t) \rho\big].\label{eq_TO2_energy}\\
\end{aligned}
\end{sequation}

\noindent Moreover, $\mathcal{Q}(\mathcal{X}^m, \mathcal{X}^c)$ is given as

\vspace{-0.8em}
{\small
\begin{alignat}{1}
& \mathcal{Q}(\mathcal{X}^m, \mathcal{X}^c)=Q_m^{\text{E}}(t)\big[\sum_{u=1}^{U_m(t)}(x_u^m(t)-x_u^{m \rightarrow c}(t)) s_u(t)-\frac{f_m \delta}{\rho} \big] \notag\\
    & +Q_m^{\text{C}}(t)\big[\sum_{u=1}^{U_m(t)} x_u^{m \rightarrow c}(t) s_u(t)-r_m^c  \delta\big] \notag\\
    & +Z_{m}^{\text{E}}(t)\frac{Q_m^{\text{E}}(t)}{\frac{1}{t}\big[\sum\limits_{i=1}^{t-1} A_m^{\text{E}}(i)+\sum\limits_{u=1}^{U_m(t)}(x_u^m(t)-x_u^{m \rightarrow c}(t)) s_u(t)\big]}\notag\\
    & +Z_{m}^{\text{C}}(t)\frac{Q_m^{\text{C}}(t)}{\frac{1}{t}\big[\sum\limits_{i=1}^{t-1} A_m^{\text{C}}(i)+\sum\limits_{u=1}^{U_m(t)} x_u^{m \rightarrow c}(t) s_u(t)\big]}.
\end{alignat}
}

\par Correspondingly, the energy consumption minimization problem for UDs can be expressed as

\vspace{-0.8em}
{\small
\begin{align}
\mathbf{\bar{P}^{\prime\prime\prime}.2.1}:\,& \min_{{\mathcal{X}^m}} \mathcal{E} (\mathcal{X}^m, \mathcal{A}^*) \label{eq_pro_energy}\\
&\text {s.t. }
 \eqref{TOa}. \notag
\end{align}}

\noindent Similarly, the  queue stability subproblem for MEC server can be given as

\vspace{-0.8em}
{\small
\begin{align}
\mathbf{\bar{P}^{\prime\prime\prime}.2.2}:\,& \min_{\mathcal{X}^m,\mathcal{X}^c} \mathcal{Q}(\mathcal{X}^m, \mathcal{X}^c) \label{eq_pro_que}\\
&\text {s.t. }
 \eqref{TOb}\text{ and } \eqref{TOc}. \notag
\end{align}}

\par In following, we design a two-stage alternating optimization method for task offloading subproblem. 

\textbf{Energy Minimization Subproblem.} We first focus on solving the energy minimization subproblem $\mathbf{\bar{P}^{\prime\prime\prime}.2.1}$. However, solving this subproblem directly remains difficult since it is an ILP problem, as proved in Theorem \ref{lem_P21_ILP}.

\begin{theorem}
\label{lem_P21_ILP}
Problem $\mathbf{\bar{P}^{\prime\prime\prime}.2.1}$ is an NP-hard integer linear programming (ILP) problem.
\end{theorem}

\begin{proof}
The details can be found in Appendix \ref{app_lem_P21_INLP} of the supplemental material. 
\end{proof}
\vspace{-0.5em}

\par This subproblem has the following features. First, from the perspective of UDs, each UD focuses on minimizing its own energy consumption. However, from the perspective of MEC servers, each MEC server manages the energy consumption of its task queue with a set of tasks that are offloaded from different UDs. Accordingly, the UDs and MEC servers have heterogeneous preferences for the energy minimization subproblem. Moreover, in the considered MEC system, different UDs have diverse processing requirements for various tasks, while different MEC servers have varying computing capabilities. This results in each UD having heterogeneous preferences on different MEC servers when offloading tasks, and each MEC server having heterogeneous preferences on different UDs when providing services. Therefore, the efficient task offloading decision of each UD should guarantee the UD is associated with a suitable MEC server.

\par Considering that the bilateral matching mechanism provides an effective tool for establishing mutually beneficial relationship between two sets of entities with heterogeneous preferences, we employ this mechanism to facilitate the task offloading between UDs and MEC servers~\cite{liu2022task}. First, the bilateral matching allows for more efficient task offloading by considering the preferences and capabilities of both UDs and MEC servers~\cite{liu2024coexistence}. This guarantees that tasks are offloaded to the most suitable destinations. Moreover, by incorporating preferences of both UDs and MEC servers, the matching process can achieve a more balanced distribution of tasks at the edge. This prevents the task overloading or resource idle at certain MEC servers, thereby enhancing the load balance and system stability.

\par Based on the abovementioned analysis, we employ the bilateral matching mechanism to solve the energy minimization subproblem. 

\par \textbf{First}, the bilateral matching is defined in Definition~\ref{def_match}.

\begin{definition}
	\label{def_match}
	 The bilateral matching is defined as  $(\mathcal{A},\eta)$, where 
	\begin{itemize}[]
		\item $\mathcal{A}=(\mathcal{U}_m(t),\mathcal{I})$ comprises two distinct sets of agents, i.e., UDs $\mathcal{U}_m(t)$ and servers $\mathcal{I}$.Specifically, since each UD can select to execute the task locally, the UD can be regarded as a special local server for task processing. Therefore, the set of servers is denoted as $\mathcal I=\left\{ i_0,i_1\right\}$, where $i_0$ and $i_1$ represent the local server and MEC server, respectively.
	
		\item $\eta\subseteq \mathcal{U}_m(t) \times \mathcal{I}$ represents the matching between the UDs and the servers. Since we consider a binary offloading scheme, the matching $\eta$ is defined as a many-to-one mapping between the UDs $\mathcal{U}_m(t)$ and servers $\mathcal I$. Specifically, each UD $u\in\mathcal{U}_m(t)$ can be matched with at most one server, i.e., $\eta(u)\in \mathcal{I}$, while each server $i\in\mathcal{I}$ can be matched with multiple UDs, i.e., $\eta(i)\subseteq \mathcal{U}_i(t)$. Moreover, the matching must meet the constraints as follows:

            \vspace{-0.8em}
            {\small
          \begin{subequations}
          \label{eq_matching}
            \label{eq_virtual}
            \begin{alignat}{1}
                &\eta(u)\in \mathcal{I}, \, \left|\eta (u)\right|=1, \,  \forall u \in \mathcal{U}_m(t), \label{eq_matching_c1}\\
                &\eta(i)\subseteq \mathcal{U}_m(t), \, \left|\eta (i)\right| \le q_i, \,  \forall i \in \mathcal I, \, \label{eq_matching_c2}\\
                &\eta(u) = i  \iff u \in \eta(i), \, \forall u \in \mathcal{U}_m(t), \, i \in \mathcal I, \label{eq_matching_c3}
            \end{alignat}
        \end{subequations}}
  \end{itemize}
\end{definition}

\noindent where \eqref{eq_matching_c1} indicates that each UD $u\in\mathcal{U}_m(t)$ can be matched with at most one server. Furthermore,  \eqref{eq_matching_c2} means that each server $i$ can be matched with multiple UDs, which is constrained by the matching capacity $q_i$. Here, the matching capacity of the local server is defined as the number of UDs connected to MEC server $m$, i.e., $q_{i_0} = |\mathcal{U}_m(t)|$. Additionally, \eqref{eq_matching_c3} implies that if UD $u$ is matched with server $i$, server $i$ provides computing service for UD $u$, and vice versa.

\par \textbf{Second}, as discussed before, the UDs and servers have heterogeneous preferences on each other, which should be measured mathematically. To this end, the utility function is defined to estimate the preferences of UDs and servers. Specifically, since the goal of subproblem $\mathbf{\bar{P}^{\prime\prime\prime}.2.1}$ is to minimize the UD energy consumption, we model the utility functions of each UD $u$ and server $i$ as the negative of the UD energy consumption and the negative of the total UD energy consumption, respectively, as follows:

\vspace{-0.8em}
{\small
\begin{subequations}
    \label{eq_virtual}
    \begin{alignat}{1}
        &\varphi_{u}(\eta) =  -\big ( {E_{u}^{\text{loc}}(t)+E_{u}^{\text{off}}(t)} \big ), \, \forall u \in \mathcal{U}_m(t), \label{eq_UD_utility}\\
        & \varphi_i(\eta)=\sum_{u \in \mathcal{U}}-\mathbb{I}_{\left\{i=i_0\right\}} E_u^{\text{loc}}(t) +\sum_{u \in \mathcal{U}}-\mathbb{I}_{\left\{i \neq i_0\right\}}E_u^{\text{off}}(t), \forall i \in \mathcal{I}, \label{eq_server_utility}     
    \end{alignat}
\end{subequations}

\noindent where the matching preferences of UDs and servers are based on the utility functions in a descending order.
}

\par \textbf{Finally}, we introduce the main steps of the bilateral matching, which mainly consists of the steps of matching initialization, stable matching, and result mapping. 

\par \textit{Matching Initialization.}  First, a matching $\eta$ is initialized randomly. Then, the utility functions are calculated based on Eqs. \eqref{eq_UD_utility} and \eqref{eq_server_utility} to quantify the matching preferences of UDs and servers.

\par \textit{Matching Process.}  Effective bilateral matching enables a stable result across all pairs of agents. However, the communication resource allocation for a UD is affected by the matching results of the other UDs, which in turn impacts the estimation of the utility functions. This interdependence, referred to as externality~\cite{TD2022}, implies that the preference relationships among UDs and servers may change during the matching process, which potentially influence the stability of the matching result. To address the externality, we adopt the \textit{swap matching}, which is defined as follows.

\begin{definition}
\label{def_swap}
Swap Matching: For a pair of UD $(i,i^{\prime})$ and a pair of server $(u, u^{\prime})$ in  matching $\eta$, where UD $u$ and server $i$ is matched, i.e., $i=\eta(u)$, $i^{\prime}=\eta(u^{\prime})$, and $\forall i\neq i^{\prime}$, $\forall u\neq u^{\prime}$, a swap matching $\eta_{ui}^{u^\prime i^\prime}$ can be given as
\begin{sequation}
\label{eq_swap_matching}
\eta_{ui}^{u^\prime i^\prime}=\left\{\eta \backslash\left\{(u, i),\left(u^{\prime}, i^{\prime}\right)\right\}\right\} \cup\left\{\left(u, i^{\prime}\right),\left(u^{\prime}, i\right)\right\},
\end{sequation}

\noindent where $i^{\prime}=\eta(u)$ and $i=\eta(u^{\prime})$, and the matching state in $\eta_{ui}^{u^\prime i^\prime}$ for the other UDs and servers remain unchanged. 
\end{definition}

\par Based on Definition \ref{def_swap}, we can known that a swap matching can be performed to change the interdependent relationships in the current matching. Therefore, we present the criteria of operating the swap matching by introducing the concept of a \textit{blocking pair} as follows.

\begin{definition}
\label{def_block}
Blocking Pair: For a pair of UD $(i,i^{\prime})$ and a pair of server $(u, u^{\prime})$ in  matching $\eta$, where UD $u$ and server $i$ is matched, i.e., $i=\eta(u)$, $i^{\prime}=\eta(u^{\prime})$, and $\forall i\neq i^{\prime}$, the UD pair $(i,i^{\prime})$ is a blocking pair if it satisfies the conditions as follows:

\vspace{-0.8em}
{\small
\begin{subequations}
    \label{eq_block}
    \begin{alignat}{1}
          &\forall s \in \left \{u,u^{\prime},i,i^{\prime} \right \}, \varphi_s  (\eta_{u i}^{u^{\prime} i^{\prime}})\ge \varphi_s  (\eta), \label{eq_block_1}\\
          &\exists s \in \left \{u,u^{\prime},i,i^{\prime} \right \}, \varphi_s  (\eta_{u i}^{u^{\prime} i^{\prime}})> \varphi_s (\eta),\label{eq_block_2}
    \end{alignat}
\end{subequations}}

\noindent where \eqref{eq_block_1} means that the UDs $u,u^{\prime}$ exchange their currently matched servers $i,i^{\prime}$ if their utility functions do not deteriorate, and \eqref{eq_block_2} indicates that the utility function for one UD improves as a result of the swap.

\end{definition}

\par By applying the swap matching action according to Definition \ref{def_swap}, a stable matching result is achieved when it satisfies the condition in Definition \ref{def_stable}.

\begin{definition}
\label{def_stable}
Stable Matching: A matching $\eta$ is considered as stable if there are no blocking pairs within it.
\end{definition}

\par \textit{Result Mapping.}  Based on the  matching result $\eta$, the offloading decisions at the UD side can be determined as
\begin{sequation}
\label{eq_trans_form}
x_{u}^{m}=\left\{\begin{array}{l}
0, \text { if }\eta(u) = i_0, \, \forall u \in \mathcal{U}_m(t), \\
1, \text { if }\eta(u) = i_1. \, \forall u \in \mathcal{U}_m(t)\textcolor{b}{.}
\end{array}\right.
\end{sequation}

\textit{Result Refinement.} To further reduce the energy consumption, we introduce a \textit{remove action strategy} to refine the initial matching results by removing the elements from $(\mathcal{X}^m)^*$, which is defined in Definition \ref{def_remove} as follows:

\begin{definition}
\label{def_remove}
\par Remove Action: For the given task offloading decision at the UD side $\mathcal{X}^*$ and an element $x\in \mathcal{X}^*$, a new offloading decision $\mathcal{X}^{\prime}$ is obtained by removing the element $x$, i.e., $\mathcal{X}^{\prime}=\mathcal{X}^*\setminus \left \{x  \right \}$, where $\mathcal{E}(\mathcal{X}^{\prime})\le \mathcal{E}(\mathcal{X}^*)$.
\end{definition}

\par The solution process of energy minimization subproblem is presented in Algorithm \ref{algo_matching}. In the matching initialization phase, a matching $\eta$ is randomly initialized, and the utility functions for UDs and servers are calculated (lines 1 to 2). Then, in the stable matching phase, the swap matching action is performed, and the utility functions are updated each time a blocking pair is identified (lines 3 to 7). The swap matching is terminated until there are no blocking pairs exist in the current matching, resulting in a stable matching result. Additionally, the matching result is mapped into the offloading decisions at the UD side based on Eq. \eqref{eq_trans_form} (line 9). Finally, the matching result is refined by employing the remove action (lines 10 to 14).

 \vspace{-0.5em}

 \begin{algorithm}[]
\caption{Energy Minimization Subproblem}  
\label{algo_matching}	
 \SetAlgoLined
	\KwIn{ $q_{i1}$}
    \KwOut {$(\mathcal{X}^m)^*$}
    \tcp{Matching Initialization}
    Randomly initialize a matching $\eta$\;
   Obtain the utility functions based on Eqs. \eqref{eq_UD_utility} and \eqref{eq_server_utility}\; 
    \tcp{Matching Process}
     \Repeat{$\text{There are no blocking pairs in the current matching.} \eta$}{
     Select a pair of UDs $(u,u^{\prime})$, where $\eta(u)=i$ and $\eta(u^{\prime})=i^{\prime}$\;
    \If{$(u,u^{\prime})$ is a blocking pair of the current matching $\eta$}{
        $\eta = \eta_{ui}^{u^{\prime}i^{\prime}}$\;
        Update utility functions based on Eqs. \eqref{eq_UD_utility} and \eqref{eq_server_utility}\;
    }
 }
    \tcp{Result Mapping}
    Map the matching result $\eta$ into the offloading decisions at UD side based on Eq. \eqref{eq_trans_form}\;
    \tcp{Result Refinement}
     \Repeat{There is no remove action can be executed}{ 
       Select an element $x \in (\mathcal{X}^m)^*$\; 
         \If {$\mathcal{E}\left (\mathcal{X}^*\setminus \left \{x  \right \}\right )\le \mathcal{E}((\mathcal{X}^m)^*)$}{
         $(\mathcal{X}^m)^*=\left (\mathcal{X}^*\setminus \left \{x  \right \}\right)$\;
        }
     }
\end{algorithm}
\vspace{-1em}

\par \textbf{Queue Stability Subproblem.}
By fixing the offloading decisions at the UD side $(\mathcal{X}^m)^*$, we can obtain the UDs who offload their tasks to  the MEC server $m$, denoted by $\mathcal U_m^{\text{O}}=\left\{u \in \mathcal{U}_m(t) \mid x_{u}^{m}(t)=1\right\}$. Then, the queue stability subproblem in \eqref{eq_pro_que} is transformed as follows: 

\vspace{-0.8em}
{\small
\begin{align}
\mathbf{\bar{P}^{\prime\prime\prime}.2.2^{\prime}}:& \min _{\{x_{u}^{m \rightarrow c}\}}Q_m^{\text{E}}(t)\big(\sum_{u \in \mathcal U_m^{\text{O}}}(1-x_u^{m \rightarrow c}(t)) s_u(t)\big) \nonumber \\
& +Q_m^{\text{C}}(t)\big(\sum_{k \in \mathcal U_m^{\text{O}}} x_u^{m \rightarrow c}(t) s_u(t)\big) \nonumber\\
& +Z_{m}^{\text{E}}(t)\frac{Q_m^{\text{E}}(t)}{\frac{1}{t}\Big[\sum\limits_{i=1}^{t-1} A_m^{\text{E}}(i)+\sum\limits_{k \in \mathcal U_m^{\text{O}}}(1-x_u^{m \rightarrow c}(t)) s_u(t)\Big]}\nonumber\\
& +Z_{m}^{\text{C}}(t)\frac{Q_m^{\text{C}}(t)}{\frac{1}{t}\Big[\sum\limits_{i=1}^{t-1} A_m^{\text{C}}(i)+\sum\limits_{u \in \mathcal U_m^{\text{O}}} x_u^{m \rightarrow c}(t) s_u(t)\Big]}\label{P8}\\
\text {s.t. }&x_u^{m \rightarrow c}(t) \in\{0,1\},\forall u \in \mathcal U_m^{\text{O}}\tag{\ref{P8}{a}},\label{P8a}\\
& x_u^{m \rightarrow c}(t)\le (x_u^{m}(t))^*,\forall u \in \mathcal U_m^{\text{O}}\tag{\ref{P8}{b}}\label{P8b}.
\end{align}}

\par Problem $\mathbf{\bar{P}^{\prime\prime\prime}.2.2^{\prime}}$ is an INLP problem due to the binary variable of $x_{k}^{m \rightarrow c}$. To simplify the problem solving process, we convert the binary constraint (\ref{P8}{a}) into a continuous constraint. Specifically, by relaxing the binary decision variables $x_{k}^{m \rightarrow c}$ to be fractional, problem $\mathbf{\bar{P}^{\prime\prime\prime}.2.2^{\prime}}$ is converted as follows:

\vspace{-0.8em}
{\small
\begin{align}\mathbf{\bar{P}^{\prime\prime\prime}.2.2^{\prime\prime}}:& \min _{\{x_{u}^{m \rightarrow c}\}}Q_m^{\text{E}}(t)\big(\sum_{u \in \mathcal U_m^{\text{O}}}(1-x_u^{m \rightarrow c}(t)) s_u(t)\big) \nonumber \\
& +Q_m^{\text{C}}(t)\big(\sum_{u \in \mathcal U_m^{\text{O}}} x_u^{m \rightarrow c}(t) s_u(t)\big) \nonumber\\
& +Z_{m}^{\text{E}}(t)\frac{Q_m^{\text{E}}(t)}{\frac{1}{t}\big[\sum\limits_{i=1}^{t-1} A_m^{\text{E}}(i)+\sum\limits_{u \in \mathcal U_m^{\text{O}}}(1-x_u^{m \rightarrow c}(t)) s_u(t)\big]}\nonumber\\
& +Z_{m}^{\text{C}}(t)\frac{Q_m^{\text{C}}(t)}{\frac{1}{t}\big[\sum\limits_{i=1}^{t-1} A_m^{\text{C}}(i)+\sum\limits_{u \in \mathcal U_m^{\text{O}}} x_u^{m \rightarrow c}(t) s_u(t)\big]}\label{P9}\\
\text {s.t. }&x_u^{m \rightarrow c}(t) \in [0,1],\forall u \in \mathcal U_m^{\text{O}} \tag{\ref{P9}{a}}\\
&\eqref{P8a} \text { and }\eqref{P8b}\notag.
\end{align}}

\par Given that the objective function and constraint \eqref{P9} are convex, we can apply existing optimization tools such as CVX to derive the fractional solution for $\mathbf{\bar{P}^{\prime\prime\prime}.2.2^{\prime\prime}}$.

\par Next, we transform the fractional solution of problem $\mathbf{\bar{P}^{\prime\prime\prime}.2.2^{\prime\prime}}$ into an integral solution by employing the dependent rounding method as a straightforward solution since it can effectively convert continuous solutions into feasible binary solutions while achieving near-optimal performance with lower computational overhead~\cite{gandhi2006dependent}. The rounding operation of the dependent rounding method is presented as follows.

\par  First, given the fractional solution $\overline{\mathcal{X}^{c}}$ of problem \eqref{P9}, the fractional solutions of two arbitrary UDs $u_1,u_2\in \mathcal U_m^{\text{O}}$ are selected, i.e., $x_{u1}^{m \rightarrow c},x_{u2}^{m \rightarrow c} \in \overline{\mathcal{X}^{c}}$. Then, for these fractional solutions, we introduce the associated weight coefficients $\varpi _{u1}^{m \rightarrow c}$ and $\varpi _{u2}^{m \rightarrow c}$, which are initialized as the task size of UDs $u_1$ and $u_2$, respectively, i.e., $\varpi _{u1}^{m \rightarrow c}=s_{u1}(t)$ and $\varpi _{u2}^{m \rightarrow c}=s_{u2}(t)$.

\par Additionally, we define two parameters $\iota_1$ and $\iota_2$ as

\vspace{-0.8em}
{\small
\begin{subequations}
    \label{eq_parameter_varpi}
    \begin{alignat}{1}
          &\iota_1 = \min \{1-x_{u1}^{m \rightarrow c},\frac{\varpi _{u2}^{m \rightarrow c}}{\varpi _{u1}^{m \rightarrow c}}x_{u2}^{m \rightarrow c}\},\\
          &\iota_2 = \min \{x_{u1}^{m \rightarrow c},\frac{\varpi _{u2}^{m \rightarrow c}}{\varpi _{u1}^{m \rightarrow c}}(1-x_{u2}^{m \rightarrow c})\}.
    \end{alignat}
\end{subequations}}

\par Moreover, based on the parameters in Eq. \eqref{eq_parameter_varpi}, the probability values $\varrho_1$ and $\varrho_2$ are defined as follows:

\vspace{-0.8em}
{\small
\begin{subequations}
    \label{eq_parameter_varrho}
    \begin{alignat}{1}
          &\varrho_1 = \frac{\iota_1}{\iota_1+\iota_2}, \\
          &\varrho_2 = \frac{\iota_2}{\iota_1+\iota_2}.
    \end{alignat}
\end{subequations}}

\par Finally, based on the probability values in Eq. \eqref{eq_parameter_varrho}, the fractional solutions $x_{u1}^{m \rightarrow c}$ and $x_{u2}^{m \rightarrow c}$ are transformed into binary solutions based on the following rules. Specifically, when $\varrho_1 \ge \varrho_2$, $x_{u1}^{m \rightarrow c}$ and $x_{u2}^{m \rightarrow c}$ are updated as

\vspace{-0.8em}
{\small
\begin{subequations}
    \label{eq_frac_to_int_1}
    \begin{alignat}{1}
          &(x_{u1}^{m \rightarrow c})^{\prime}=x_{u1}^{m \rightarrow c}-\iota_2, \\
          &(x_{u2}^{m \rightarrow c})'=x_{u2}^{m \rightarrow c}+\frac{\varpi _{u1}^{m \rightarrow c}}{\varpi _{u2}^{m \rightarrow c}}\iota_2.
    \end{alignat}
\end{subequations}}

\noindent When $\varrho_1 \le \varrho_2$, $x_{u1}^{m \rightarrow c}$ and $x_{u2}^{m \rightarrow c}$ are updated as

\vspace{-0.8em}
{\small
\begin{subequations}
    \label{eq_frac_to_int_2}
    \begin{alignat}{1}
          &(x_{u1}^{m \rightarrow c})^{\prime}=x_{u1}^{m \rightarrow c}+\iota_1, \\
          &(x_{u2}^{m \rightarrow c})^{\prime}=x_{u2}^{m \rightarrow c}-\frac{\varpi _{u1}^{m \rightarrow c}}{\varpi_{u2}^{m \rightarrow c}}\iota_1.
    \end{alignat}
\end{subequations}}

\par The abovementioned process of dependent rounding method ensures that the rounded solutions satisfy the binary constraints, 
while maintaining the same size of tasks uploaded to the cloud before and after the rounding process, as proved in Theorems \ref{lem_round_1} and \ref{lem_round_2}.

\begin{theorem}
\label{lem_round_1}
In a rounding operation, at least one of the two fractional solutions are rounded to 0 or 1.
\end{theorem}

\begin{proof}
The details can be found in Appendix \ref{app_lem_round_1} of the supplemental material. 
\end{proof}
\vspace{-0.5em}
 
\begin{theorem}
\label{lem_round_2}
In a rounding operation, whether updating the fractional solutions according to Eq. \eqref{eq_frac_to_int_1} or Eq. \eqref{eq_frac_to_int_2}, it can be ensured that $x_{u1}^{m \rightarrow c}\varpi _{u1}^{m \rightarrow c}+x_{u2}^{m \rightarrow c}\varpi _{u2}^{m \rightarrow c}
=
(x_{u1}^{m \rightarrow c})^{\prime}\varpi _{u1}^{m \rightarrow c}+(x_{u2}^{m \rightarrow c})^{\prime}\varpi _{u2}^{m \rightarrow c}$.
\end{theorem}

\begin{proof}
The details can be found in Appendix \ref{app_lem_round_2} of the supplemental material. 
\end{proof}


\par Therefore, based on Lemmas \ref{lem_round_1} and \ref{lem_round_2}, by iteratively applying the rounding operation to the fractional solutions of problem \eqref{P9}, a near-optimal integer solution for problem \eqref{P8} can be obtained. We provide the process of solving problem $\mathbf{\bar{P}^{\prime\prime\prime}.2.2^{\prime}}$ in Algorithm \ref{algo_cvx_rounding}. 

In the initial phase, the fractional solution $\overline{\mathcal{X}^{c}}$ is obtained by solving problem \eqref{P9} using CVX (line 1). Following this, the rounding process begins, where the UD set with fractional offloading decisions  $\overline{\mathcal{U}_m^{c}}$, the UD set with integral offloading decisions $\widehat{\mathcal{U}_m^{c}}$, and integral offloading decision set $\widehat{\mathcal{X}^c}$ are initialized as empty (line 2). The UDs are then categorized based on the values of $x_u^{m \rightarrow c}$ into two groups, i.e., UDs with binary solutions and UDs with fractional solutions (lines 3 to 7). Next, the rounding process iterates over the UDs with fractional offloading decisions in $\overline{\mathcal{U}_m^{c}}$ (lines 8 to 18). In each iteration, two UDs $u_1$ and $u_2$ are randomly selected, and a rounding operation is performed to update their offloading decisions (lines 9 to 10). If either updated decision becomes binary (0 or 1), the corresponding UDs are moved from $\overline{\mathcal{U}_m^{c}}$ to $\widehat{\mathcal{U}_m^{c}}$, and their decisions are added to $\widehat{\mathcal{X}^c}$ (lines 11 to 16). The rounding process continues until all fractional solutions are converted to binary, yielding the near-optimal integer solution $(\mathcal{X}^c)^*$ (lines 17 to 18).

 \begin{algorithm}[]
\caption{Queue Stabilization}  
\label{algo_cvx_rounding}	
 \SetAlgoLined
	\KwIn{$(\mathcal{X}^m)^*$}
    \KwOut {$(\mathcal{X}^c)^*$}
    \tcp{Fractional Solution}
    Obtain the fractional solution $\overline{\mathcal{X}^{c}}$ of problem \eqref{P9} using CVX\;
    \tcp{Rounding Process}
     Set $\overline{\mathcal{U}_m^{c}}=\{\},\widehat{\mathcal{U}_m^{c}}=\{\},\widehat{\mathcal{X}^{c}}$=\{\}\;
    \For{each $x_{u}^{m \rightarrow c} \in \overline{\mathcal{X}^{c}}$}{
    \If{$x_{u}^{m \rightarrow c} \in \{0,1\}$}{$\widehat{\mathcal{U}_m^{c}}=\widehat{\mathcal{U}_m^{c}}\cup\{u\}$\;
    }
    \Else{
        $\overline{\mathcal{U}_m^{c}}=\overline{\mathcal{U}_m^{c}}\cup\{u\}$\;
    }
}
     \Repeat{$\text{No element exists in}$
     $\overline{\mathcal{U}_m^{c}}$}{
     Select a pair of UDs $\{u_1,u_2\}\subseteq\overline{\mathcal{U}_m^{c}}$ randomly\;
     Update $\{x_{u1}^{m \rightarrow c},x_{u2}^{m \rightarrow c}\}$ to $\{(x_{u1}^{m \rightarrow c})^{\prime},(x_{u2}^{m \rightarrow c})^{\prime}\}$ by executing rounding operation\;
    \If{$(x_{u1}^{m \rightarrow c})^{\prime}\in \{0,1\}$}{
       $\widehat{\mathcal{U}_m^{c}}=\widehat{\mathcal{U}_m^{c}}\cup\{u_1\},\overline{\mathcal{U}_m^{c}}=\overline{\mathcal{U}_m^{c}}\setminus\{u_1\}$\;
       $\overline{\mathcal{X}^{c}}=\overline{\mathcal{X}^{c}}\cup\{(x_{u1}^{m \rightarrow c})^{\prime}\}$\;
    }
    \If{$(x_{u2}^{m \rightarrow c})^{\prime}\in \{0,1\}$}{$\widehat{\mathcal{U}_m^{c}}=\widehat{\mathcal{U}_m^{c}}\cup\{u_2\},\overline{\mathcal{U}_m^{c}}=\overline{\mathcal{U}_m^{c}}\setminus\{u_2\}$\;
       $\widehat{\mathcal{X}^c}=\widehat{\mathcal{X}^c}\cup\{(x_{u2}^{m \rightarrow c})\}$\;
    }
 }
 \Return $(\mathcal{X}^c)^*=\widehat{\mathcal{X}^c}$\;
\end{algorithm}

\par \textbf{Two Stage Alternative Task Offloading Method.} To effectively address both the energy minimization and queue stability subproblems, it is crucial to integrate these two objectives and manage the trade-offs between minimizing energy consumption and ensuring queue stability. To this end, the connection capacity of the MEC server is utilized to connect the two objectives for the following reasons. Specifically, more tasks can be offloaded to MEC server $m$ as its connection capacity increases, resulting in lower energy consumption for the UDs. However, this increase also results in a rise in the queue backlog of the MEC $m$, making the system more prone to instability. Conversely, if fewer tasks are offloaded to MEC $m$, more tasks are processed locally, which increases the UD energy consumption while reducing the queue backlog of MEC server $m$. Therefore, the connection capacity of the MEC server is used to balance the queue stability for MEC servers and the energy consumption for UDs.
 
\par The two-stage alternating task offloading method is elaborated in Algorithm \ref{al_two_stage}. Initially, the temporary variable $n$ is initialized as the connection capacity of MEC server $m$, and the optimal value of task offloading problem $\mathcal{P}^{opt}$ is initialized to infinity (line 1). Next, in each iteration, the connection capacity of MEC server $m$ is set as $n$, i.e., $q_{i1}=n$ (line 3). Moreover, the offloading decisions at the UD side are obtained by solving the energy minimization subproblem via Algorithm \ref{algo_matching} (line 4). Then, the objective function value $\mathcal{E}_n$ for energy minimization subproblem is updated (line 5). Similarly, the offloading decisions at the MEC server are determined by solving the queue stability subproblem through Algorithm \ref{algo_cvx_rounding}, and the objective function value $\mathcal{Q}_n$ of the queue stability subproblem is updated (lines 6 to 7). Subsequently, the objective function value for task offloading problem is updated by comparing it with $V^{opt}$ (lines 8 to 11). After each iteration, the temporary variable $n$ is decremented by 1, and this process continues until $n$ reaches 0 (lines 12-13). Finally, the integral values of the task offloading variables are obtained using the dependent rounding algorithm (line 14).

\vspace{-0.6em}
\begin{algorithm} 
    \caption{Two-Stage Alternative Optimization Algorithm at Time Slot $t$}     
    \label{al_two_stage}
    \SetAlgoLined
    \KwIn{$C_m$ }
    \KwOut{$\mathcal{X}^{\text{opt}}(t)$}
    \textbf{initialization:} $n= C_m, {\mathcal{P}}^{\text{opt}} = \inf$\;
    \Repeat{$n=0$}{ 
     Set the maximum number of matching as $q_{i1}=n$\;
     \tcp{\text{Energy minimization subproblem}}
     Obtain ${(\mathcal{X}^m)}^*(t)$ by calling Algorithm \ref{algo_matching} for given $q_{i1}$\;
      Calculate $\mathcal{E}_n =s \mathcal{E}((\mathcal{X}^m)^*) $\;
     \tcp{\text{Queue stability subproblem}}
    Obtain $(\mathcal{X}^c)^*$ by  calling Algorithm \ref{algo_cvx_rounding} for the given $(\mathcal{X}^m)^*$\;
    Calculate $\mathcal{Q}_n = \mathcal{Q}((\mathcal{X}^m)^*,(\mathcal{X}^c)^*) $\;
     \If{$V*\mathcal{E}_n +\mathcal{Q}_n <  \mathcal{P}^{\text{opt}}$}{
     $ \mathcal{P}^{\text{opt}}= V*\mathcal{E}_n +\mathcal{Q}_n$\;
    $(\mathcal{X}^m)^{\text{opt}}(t) = (\mathcal{X}^m)^{*}(t)$\;
     $(\mathcal{X}^c)^{\text{opt}}(t) =  (\mathcal{X}^c)^{*}(t)$\;
     }
     $n \gets n-1$\;
    } 
   \Return $\mathcal{X}^{\text{opt}}(t) = \{(\mathcal{X}^m)^{\text{opt}}(t),(\mathcal{X}^c)^{\text{opt}}(t)\}$\; 
\end{algorithm}
\vspace{-1.5em}

\subsection{Main Steps of OJCTA and Performance Analysis}
\label{sec_main_step}
\par In this section, we present the main steps and performance analyses of OJCTA.

\subsubsection{Main Steps of OJCTA}
\par Algorithm \ref{al_OJCTA} presents the main steps of OJCTA. Specifically, the task queues ($Q_m^{\text{E}}(t)$ and $Q_m^{\text{C}}(t)$) and virtual queues ($Z_m^{\text{E}}(t)$ and $Z_m^{\text{C}}(t)$) are first initialized (line 1). Moreover, in each time slot, the task arrival information is acquired (line 3). Furthermore, the task offloading decision is obtained by calling Algorithm \ref{al_two_stage} (line 4), and the communication resource allocation decision is obtained according to Eq.~\eqref{eq_theo_opt_resource} (line 5). Finally, the task queues and virtual queues are updated (line 6). 

\begin{algorithm} 
\caption{OJCTA}  
\label{al_OJCTA}
    \SetAlgoLined
     \KwIn{$V$}
    \KwOut{ $\mathbf{A}^{opt}$, $\mathbf{X}^{opt}$}
    \textbf{Initialization:} Initialize the queues as $Q_{m}^{\text{E}}(1)=0, Q_{m}^{\text{C}}(1)=0, Z_{m}^{\text{E}}(1)=0, Z_{m}^{\text{C}}(1)=0$\;
    \For {Each time slot $t$}{
    Observe task arrival $I(t)$\;   
    Obtain optimal task offloading decision $\mathcal{X}^{opt}(t)$ by calling Algorithm \ref{al_two_stage}\;
    For given $\mathcal{X}^{opt}(t)$, calculate the optimal resource allocation $\mathcal{A}^{opt}(t)$ according to Eq.~\eqref{eq_theo_opt_resource}\;
    Update $Q_{m}^{\text{E}}(t+1), Q_{m}^{\text{C}}(t+1), Z_{m}^{\text{E}}(t+1), Z_{m}^{\text{C}}(t+1)$ by \eqref{eq_deadline_QL}, \, \eqref{eq_deadline_QO}, \, \eqref{eq_virtual_L}, \, \eqref{eq_virtual_O}\;
    Set $t\gets t+1$\;
   }
 \Return $\mathbf{A}^{\text{opt}} =\{\mathcal{A}^{\text{opt}}(t) \}_{t\in\mathcal{T}}$,  $\mathbf{X}^{\text{opt}} =\{\mathcal{X}^{\text{opt}}(t) \}_{t\in\mathcal{T}}$\;
\end{algorithm}
\vspace{-0.5em}

\subsubsection{Performance Analysis}
\par The computational complexity and optimality gap of OJCTA are presented in Theorems \ref{the_complexty} and \ref{the_gap}, respectively.

\begin{theorem}
\label{the_complexty}
The proposed OJCTA has a polynomial worst-case computational complexity, i.e., $\mathcal{O}(LU^3)$, where $L$ denotes the number of outer iterations in Algorithm \ref{al_two_stage}.
\end{theorem}

\begin{proof}
 The details can be found in Appendix \ref{app_the_complexty} of the supplemental material. 
\end{proof}
\vspace{-0.6em}

\begin{theorem}
\label{the_gap}
Assume that the optimal value of time-average UD energy consumption, achieved under the assumption of complete knowledge of future information, is denoted as $ C^{\text{opt}} $. Then, for any parameter $V$, the time-average UD energy consumption obtained by the  OJCTA is bounded by
\begin{sequation}
    \frac{1}{T} \sum_{t=1}^{T} \sum_{m=1}^{M} \sum_{u=1}^{U_{m}(t)} \mathbb{E}\{ E_{u}^{\text{loc}}(t) + E_{u}^{\text{off}}(t)  \}
 \le C^{\text{opt}}+\frac{B}{V}, \label{eq_teh_bound}
\end{sequation}
\noindent where $B$ is defined in Theorem \ref{the_bound_drift}.
\end {theorem}

\begin{proof}
The details can be found in Appendix \ref{app_the_gap} of the supplemental material. 
\end{proof}
\vspace{-0.5em}

\vspace{-1.2em}

\section{Simulation Results and Analysis}
\label{sec_simulation}


\subsection{Simulation Setup}
\label{simulation_set_up}

\par We first describe the settings related to the simulation experiments, including the parameters, evaluation metrics, and comparison baselines.

\subsubsection{Parameters}

\par We consider a three-layer collaborative
MEC system MEC system, where 4 MEC servers and a cloud are deployed in an area of $500 \times 500 \ \text{m}^2$ to provide computation services for UDs. Furthermore, the UDs that are covered by each MEC server is set as 50, and the connection capacity of each MEC server is set as 30. Moreover, the UDs within the coverage area of each MEC server are distributed at distances ranging from 50 m to 500 m. Additionally, the system operates over a total timeline of 200 s, with each time slot of $\delta = 1$ s. Table~\ref{tab_simuParameter} gives the other parameters.

\subsubsection{Evaluation Metrics}

\par We evaluate the overall performance of OJCTA by adopting the following indicators. i) Average UD energy consumption, i.e., $\frac{1}{T}\sum_{t=1}^T\sum_{m=1}^M\sum_{u=1}^{U_m(t)}(E_{u}^{\text{loc}}(t)+E_{u}^{\text{off}}(t))$, which indicates the average cumulative energy consumption of UDs per unit time. ii) Average queuing delay of the edge computing queue, i.e., $\frac{1}{T} \sum_{t=1}^{T} Q_{m}^{\text{E}}(t)/\tilde{A}_{m}^{\text{E}}(t)$, which represents the average delay experienced by tasks in the edge computing queue $Q_m^{\text{E}}$. iii) Average queuing delay of the cloud offloading queue, i.e., $\frac{1}{T} \sum_{t=1}^{T} Q_{m}^{\text{C}}(t)/\tilde{A}_{m}^{\text{C}}(t)$, which indicates the average delay experienced by tasks in the cloud offloading queue $Q_m^{\text{C}}$.

\subsubsection{Comparison Approaches}
\par  We evaluate the proposed OJCTA in comparison with the following baselines. 

\begin{itemize}	\item \textit{Local computing (LC)}: All UDs process their tasks locally.
	\item \textit{Random offloading (RO)}: The task offloading decisions are made randomly, without considering energy-minimization or queuing delay constraints.
	\item \textit{Energy considered first (ECF)}: The ECF approach focuses solely on minimizing the energy consumption of UDs, disregarding  the long-term queuing delay constraints.
	\item \textit{Single slot constraint approach (SSC)}\cite{CLX2018}: Instead of following long-term queuing delay constraints, the SSC approach poses hard queuing delay constraints in each time slot.
	\item \textit{No cloud collaboration (NCC)}~\cite{YHJ2021}: The NCC approach focuses on a two-tier system which consists of MEC servers and UDs, without incorporating edge-cloud collaboration.
 \item \textit{Genetic algorithm-based joint task offloading and resource allocation (GJTORA)}~\cite{zhao2018qoe}: The task offloading is determined by using the genetic algorithm, while the communication resource allocation is decided by the proposed OJCTA.
\end{itemize}

\begin{table}[t] 
	\setlength{\abovecaptionskip}{0pt}%
	\setlength{\belowcaptionskip}{0pt}%
	\caption{Simulation parameters}
	\label{tab_simuParameter}
	\renewcommand*{\arraystretch}{1}
	\begin{center}
		\begin{tabular}{p{.05\textwidth}|p{.22\textwidth}|p{.14\textwidth}}
			\hline
			\hline
			\textbf{Symbol}&\textbf{Meaning}&\textbf{Default value}\\
			  \hline
			      $p_u^{\text{tra}}$ &Transmit power of UD $u$ & $[0.1,0.5]$ W\\
			  \hline
				   $f_u$&Computing resources of UD $u$ & $[1,2]$ GHz \\
			  \hline 	     			 	$I_u$&Task size of UD $u$&  $[10^4,10^6]$ bits\\ 
			  \hline				  $d$&Distance between MEC servers and UDs&$[50,500]$ m\\
			  \hline				  $Z$&Computation intensity of tasks&1000 cycles/bit \cite{ZFH2018}\\ 
	\hline
			  $f_m$&Computing resources of MEC server $m$&5 GHz\\ 
        \hline				  $B_m$&Bandwidth between UD $u$ and MEC server $m$&20 MHz\\ 
	\hline
 $r_{m}^{c}$&Transmission data rate between MEC server $m$ and cloud $c$ &8 Mbps\\ 
	\hline
 $\zeta_u$& The effective switched capacitance coefficient of UD $u$
&$10^{-28}$\\ 
	\hline
 ${\sigma}^2$&Noise power&-98 dBm\\ 
	\hline
 $V$&Lyapunov penalty factor&$[5,40]$\\ 
	\hline
		\end{tabular}
	\end{center}
\end{table}

\subsection{Evaluation Results}
\label{Numerical Results}

\par In this section, we first assess the online offloading performance of OJCTA with default parameters. Subsequently, we compare the impacts of different parameters on the performance of OJCTA.

\subsubsection{Online Offloading Performance Evaluation}

\par Figs. \ref{fig_s1}(a), \ref{fig_s1}(b), and \ref{fig_s1}(c) evaluate the performances of average UD energy consumption, average queuing delay for edge computing queue, and average queuing delay for cloud offloading queue, respectively with time slots. 

\par From Fig. \ref{fig_s1}(a), we can see that OJCTA consistently outperforms the LC, RO, SCC, NCC, and GJTORA while underperforms the ECF with respect to the UD energy consumption over time. Several factors contribute to this outcome. First, the LC approach, where all tasks are processed locally, could lead to higher energy consumption due to the limited processing capability and battery capacity of UDs. Similarly, the random offloading of the RO approach does not prioritize energy consumption minimization, leading to suboptimal performance of the UD energy efficiency. Moreover, the SSC approach enforces strict queuing delay constraints in each time slot, which could cause workload backlogs in the long term. In this case, more tasks will be processed locally, leading to increased UD energy consumption. Furthermore, the NCC method, which omits cloud collaboration, restricts task offloading options to MEC servers alone, limiting its energy-saving potential. Besides, the genetic algorithm adopted by GJTORA often faces slow convergence and high computational demands, which can lead to suboptimal offloading decisions, especially in real-time scenarios. In addition, while the ECF approach naturally achieves lower energy consumption than the proposed OJCTA by prioritizing energy efficiency, this comes at the expense of significant queuing delays and instability, as evidenced in Figs. \ref{fig_s1}(b) and \ref{fig_s1}(c).

\par We can observe from Figs. \ref{fig_s1}(b) and \ref{fig_s1}(c) that the proposed OJCTA exhibits moderate and stable average queuing delay among the seven approaches. This is mainly because the proposed OJCTA can effectively balance the trade-off between the queuing delay and the UD energy consumption by leveraging the Lyapunov optimization framework. Specifically, in Fig. \ref{fig_s1}(b), it is evident that OJCTA maintains a lower queuing delay compared to the approaches such as LC and RO, which either overload local UDs or inefficiently manage task offloading. This demonstrates that the proposed OJCTA not only reduces UD energy consumption but also controls the backlog at the MEC server, ensuring smoother and more stable task processing. Additionally, in Fig. \ref{fig_s1}(c), OJCTA demonstrates competitive queuing delays for cloud offloading queue, which benefits from the edge-cloud collaboration mechanism in our approach. By offloading tasks dynamically from the MEC servers to the cloud, OJCTA avoids overloads at the edge, thus reducing the queuing delays compared to the approaches like NCC that do not leverage cloud collaboration. 

\par In conclusion, the simulation results in Fig. \ref{fig_s1} demonstrate that the proposed OJCTA can effectively reduces UD energy consumption while constraining the queuing delays and ensuring queuing stability at the edge.

\begin{figure*}[!hbt] 
	\centering
	\setlength{\abovecaptionskip}{1pt}%
	\setlength{\belowcaptionskip}{1pt}%
	\subfigure[Average UD energy consumption]
	{
		\begin{minipage}[t]{0.31\linewidth}
			\centering
			\includegraphics[scale=0.23]{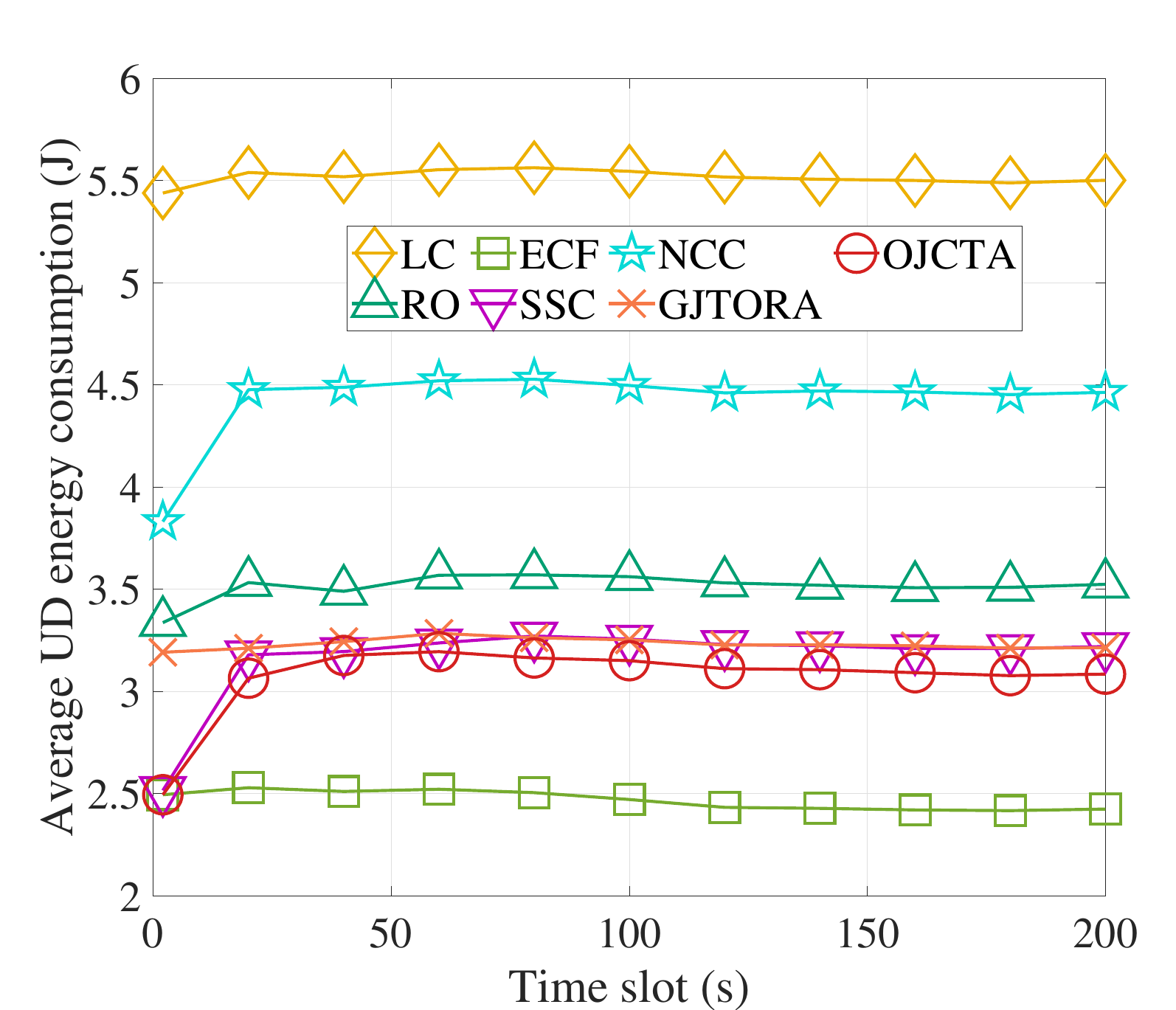}
		\end{minipage}
	}
	\subfigure[Average queuing delay of $Q^{\text{L}}$]
	{
		\begin{minipage}[t]{0.31\linewidth}
			\centering
			\includegraphics[scale=0.23]{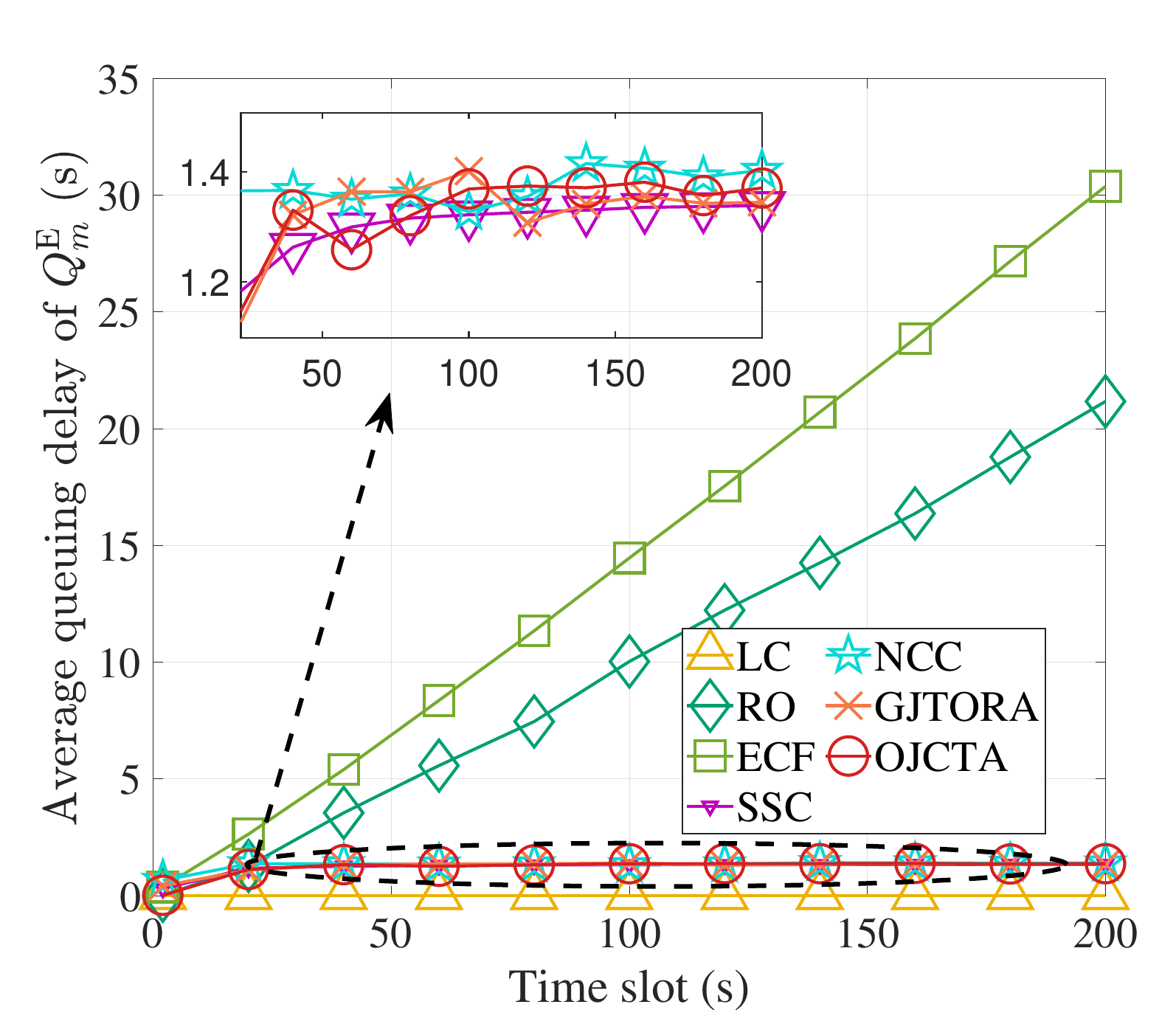}
		\end{minipage}
	}
	\subfigure[Average queuing delay of $Q^{\text{O}}$]
	{
		\begin{minipage}[t]{0.31\linewidth}
			\centering
			\includegraphics[scale=0.23]{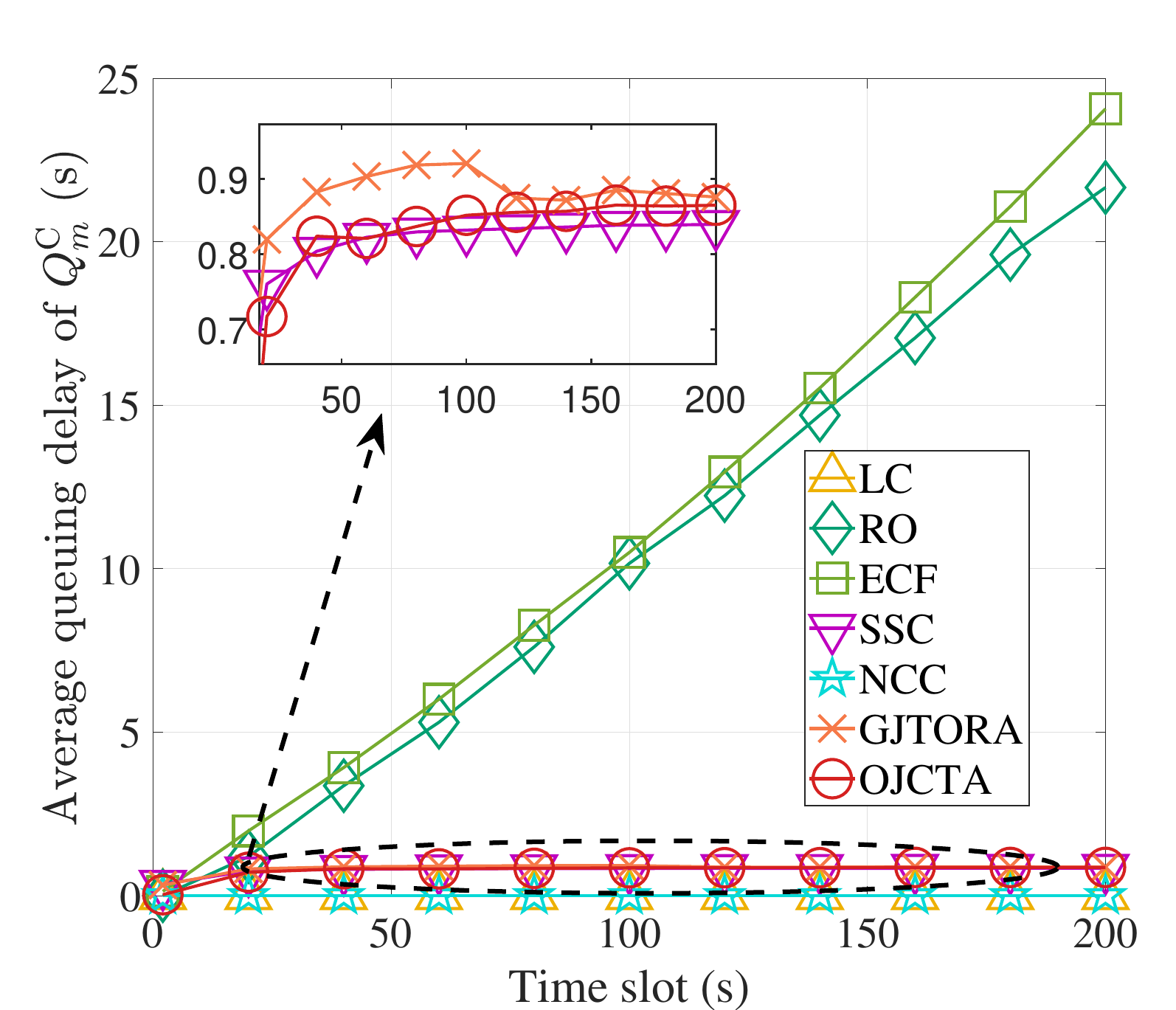}
		\end{minipage}
	}
	\caption{System performance with time slots.} 
	\label{fig_s1}
	\vspace{-1.6em}
\end{figure*}

\begin{figure*}[!hbt] 
	\centering
        \setlength{\abovecaptionskip}{1pt}%
	\setlength{\belowcaptionskip}{1pt}%
	\subfigure[Average UD energy consumption]
	{
            \begin{minipage}[t]{0.31\linewidth}
			\centering
			\includegraphics[scale=0.23]{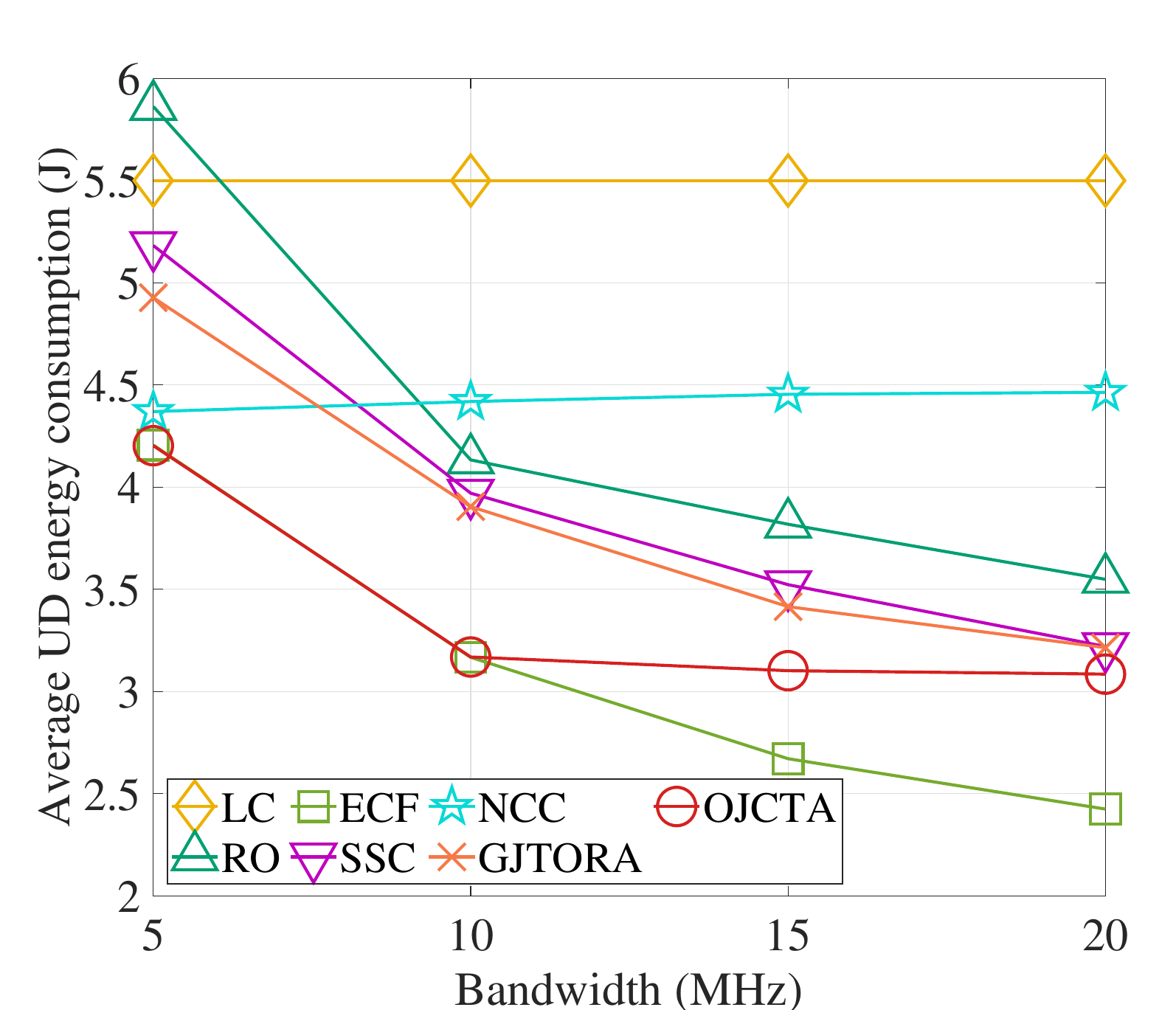}
		\end{minipage}
	}
	\subfigure[Average queuing delay of $Q^{\text{L}}$]
	{
             \begin{minipage}[t]{0.31\linewidth}
			\centering
			\includegraphics[scale=0.23]{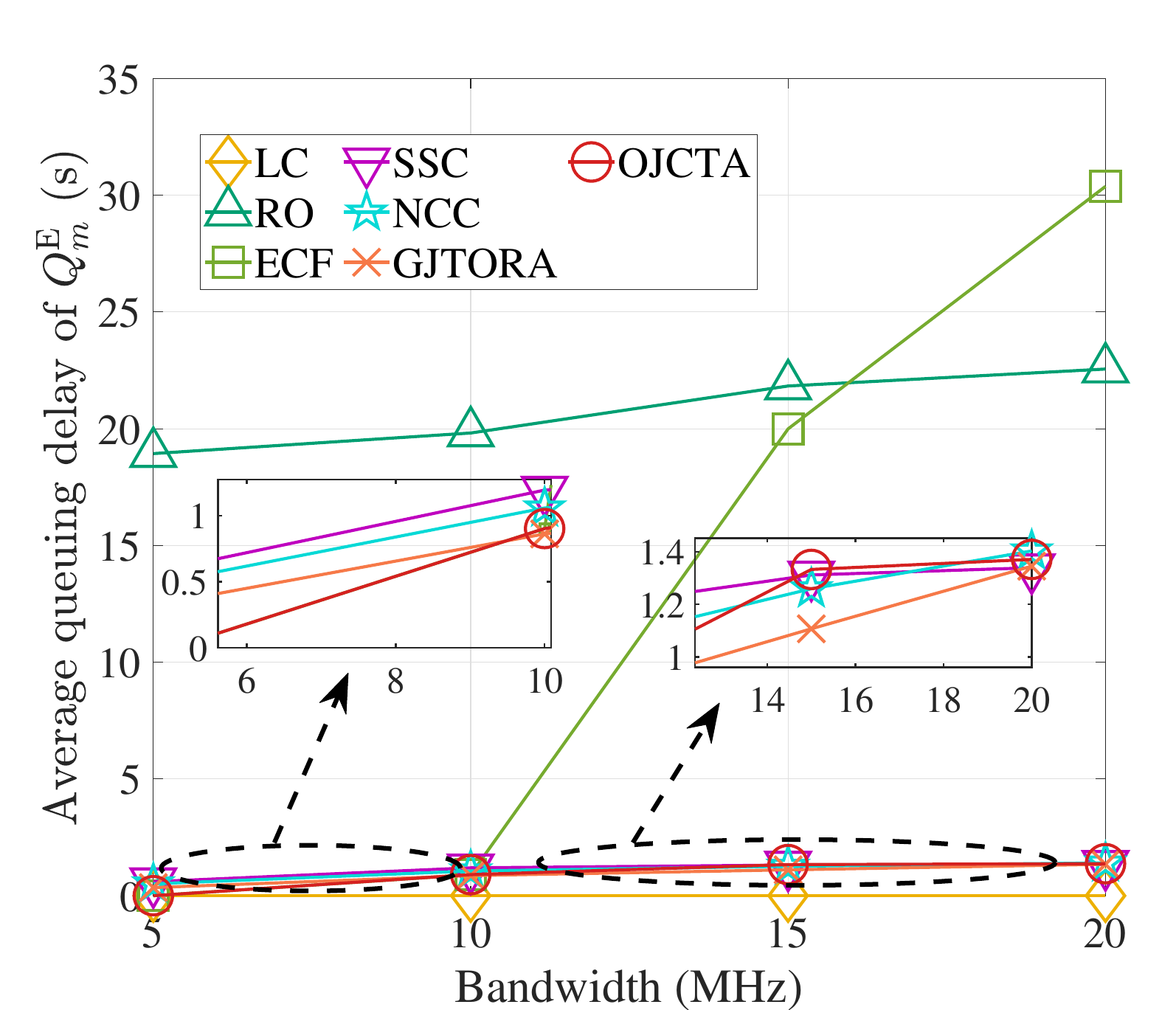}
		\end{minipage}
	}
	\subfigure[Average queuing delay of $Q^{\text{O}}$]
	{
              \begin{minipage}[t]{0.31\linewidth}
			\centering
			\includegraphics[scale=0.23]{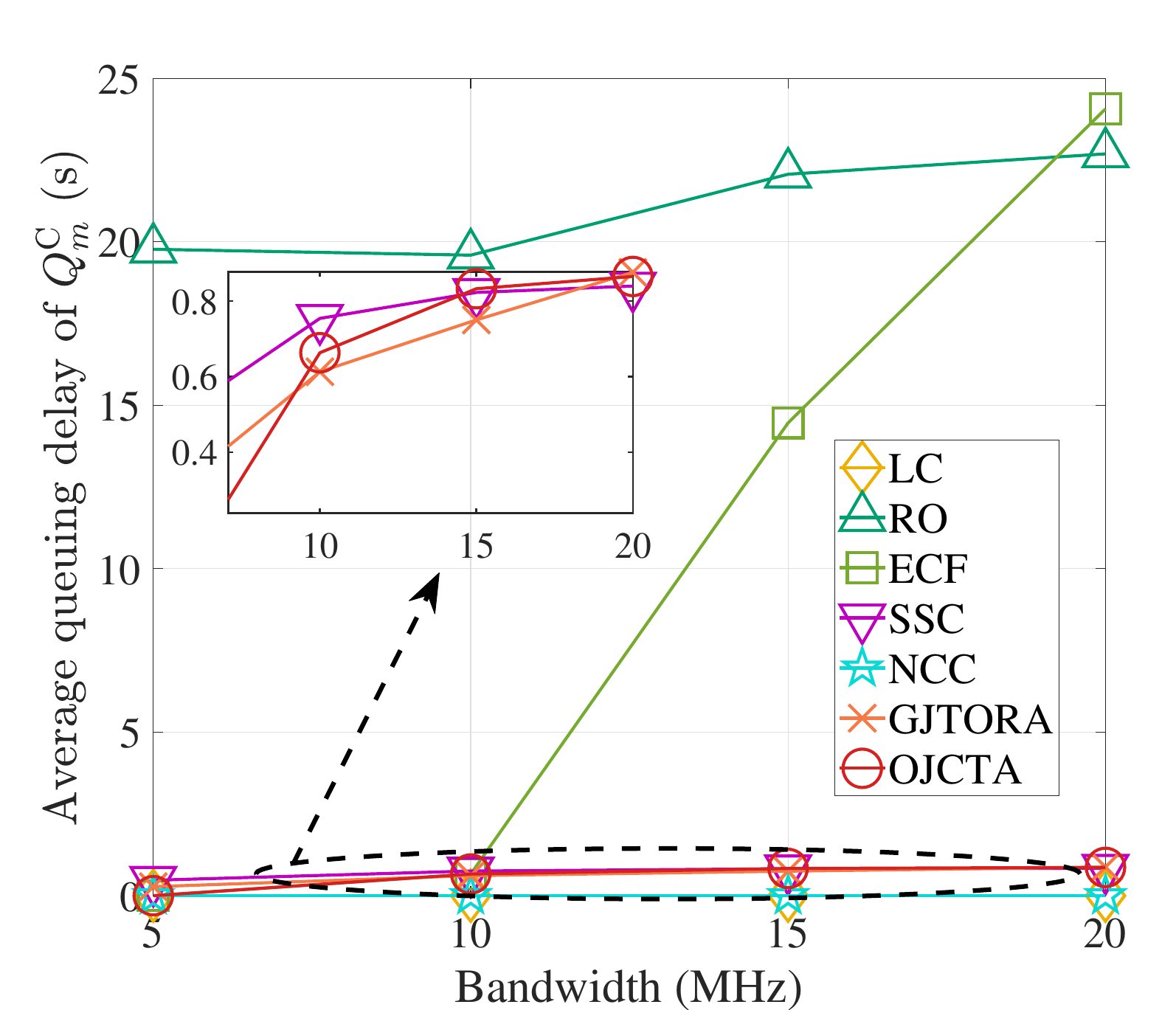}
		\end{minipage}
	}
	\centering
	\caption{System performance with bandwidth.}
	\label{fig_s3}
	\vspace{-1.6em}
\end{figure*}

\subsubsection{Impact of Parameters}
\par We compare the impacts of different system parameters on the performance of the proposed OJCTA and the benchmark approaches in this section.

\par \textbf{Impact of Bandwidth.} Figs. \ref{fig_s3}(a), \ref{fig_s3}(b) and \ref{fig_s3}(c) show the impact of bandwidth on the average energy consumption, average queuing delay for edge computing queue, and average queuing delay for cloud offloading queue, respectively. 

\par Fig. \ref{fig_s3}(a) shows that as bandwidth increases, the average UD energy consumption decreases for most approaches. This is because as the communication resources become abundant, more tasks can be offloaded to the MEC servers, thereby reducing the UD energy consumption. However, it is noteworthy that the energy consumption of LC and NCC remains nearly constant, even as bandwidth increases. The reason is that LC processes all tasks locally, and thus its energy consumption is primarily determined by the limited battery capacity and computational capabilities of UDs. Similarly, NCC restricts task offloading to only MEC servers without utilizing cloud resources, thus the energy-saving potential is constrained, leading to minimal variation in energy consumption as bandwidth increases. Moreover, we can also observed that the proposed OJCTA consistently ex superior performance compared to LC, RO, SSC, NCC, and GJTORA with respect to average energy consumption. The reason is that OJCTA can dynamically decide efficient offloading decisions, taking advantage of the available bandwidth according to the dynamic communication links. Furthermore, the ECF approach, which focuses solely on minimizing energy consumption, slightly outperforms OJCTA with respect to the UD energy consumption. However, this leads to significant queuing delay and instability since ECF does not account for queuing constraints.

\par From Figs. \ref{fig_s3}(b) and \ref{fig_s3}(c), we can observe that the queuing delays for the seven approaches generally exhibit an upward trend as bandwidth increases. This is because as the communication resources increase, more tasks could be offloaded to MEC servers, leading to heavier workloads and longer queuing delays. Moreover, compared to the benchmark approaches, the proposed OJCTA maintains moderate and relatively controlled queuing delays across different bandwidth levels. The reasons are as follows. First, the edge-cloud-collaborative architecture ensures that the queuing delays do not significantly escalate, even more tasks are offloaded with increasing bandwidth. Moreover, the Lyapunov optimization framework has the ability to guarantee the long-term queuing stability. Additionally, the close-form decision of the communication resource allocation can dynamically achieve the near-optimal result based on the available bandwidth. Besides, the two-stage alternating task offloading method effectively balances both UD energy minimization and queue stability. In contrast, approaches such as RO and ECF experience considerable queuing delays, as they fail to account for queuing delay constraints effectively.

\par In conclusion, the simulation results indicate that the proposed OJCTA successfully manages the trade-off between queuing delay and energy consumption. It achieves a balanced performance by using a dynamic task offloading strategy that adapts to varying bandwidth conditions while maintaining queuing stability.

\par \textbf{Impact of MEC Server Connection Capacity.} Figs. \ref{fig_s4}(a), \ref{fig_s4}(b), and \ref{fig_s4}(c) depict the impact of the MEC server connection capacity on the average UD energy consumption, average queuing delay for edge computing queue, and average queuing delay, respectively. 

\par  Fig. \ref{fig_s4}(a) indicates that as the MEC connection capacity of MEC servers increases, the average UD energy consumption decreases across all approaches. This reduction occurs because more tasks are offloaded to the MEC server, reducing the UD energy consumption required for local processing. Moreover, the proposed OJCTA outperforms the LC, RO, ECF, SSC, NCC, and GJTORA, consistently achieving lower energy consumption. This is because the task offloading and resource allocation in the Lyapunov optimization framework can dynamically adjusts decisions based on the connection condition. Notably, the UD energy consumption of ECF remains lower than the other approaches, as its primary goal is to minimize energy consumption. However, as we can see in Figs. \ref{fig_s4}(b) and \ref{fig_s4}(c), this comes at the cost of higher queuing delays.

\par In Figs. \ref{fig_s4}(b) and \ref{fig_s4}(c), the average queuing delays of RO and ECF for both edge computing queues and the cloud offloading queues show an initial rise followed by a significant decline. This can be attributed to the absence of long-term queuing delay constraints in these approaches. Therefore, as the number of connected UDs increases, the system initially experiences congestion due to the increased task offloading, resulting in longer queuing delays. However, once a certain threshold is reached, fewer tasks are offloaded, leading to a reduction in queuing delays. In contrast, the proposed OJCTA maintains stable and relatively low queuing delays as the connection capacity increases. This is because as more UDs are connected, OJCTA can adjust the decisions of resource allocation and task offloading, avoiding excessive queuing delays from escalating while meeting the long-term queuing constraint. Besides, the edge-cloud collaboration of OJCTA enables balanced task distribution, ensuring that neither the edge nor the cloud resources become overwhelmed, thus leading to stable queuing performance for both the cloud offloading queues and the edge computing queues.

\par In summary, the simulation results highlight the effectiveness of OJCTA in managing both energy consumption and queuing delays as the MEC server connection capacity increases. Furthermore, while the benchmark approaches suffer from high queuing delays due to congestion or task rejection, the dynamic and adaptive strategy of the proposed OJCTA ensures efficient resource utilization and task offloading, leading to lower energy consumption and stable queuing performance.

\begin{figure*}[!hbt] 
	\centering
        \setlength{\abovecaptionskip}{2pt}%
	\setlength{\belowcaptionskip}{2pt}%
	\subfigure[Average UD energy consumption]
	{
            \begin{minipage}[t]{0.31\linewidth}
			\centering
			\includegraphics[scale=0.23]{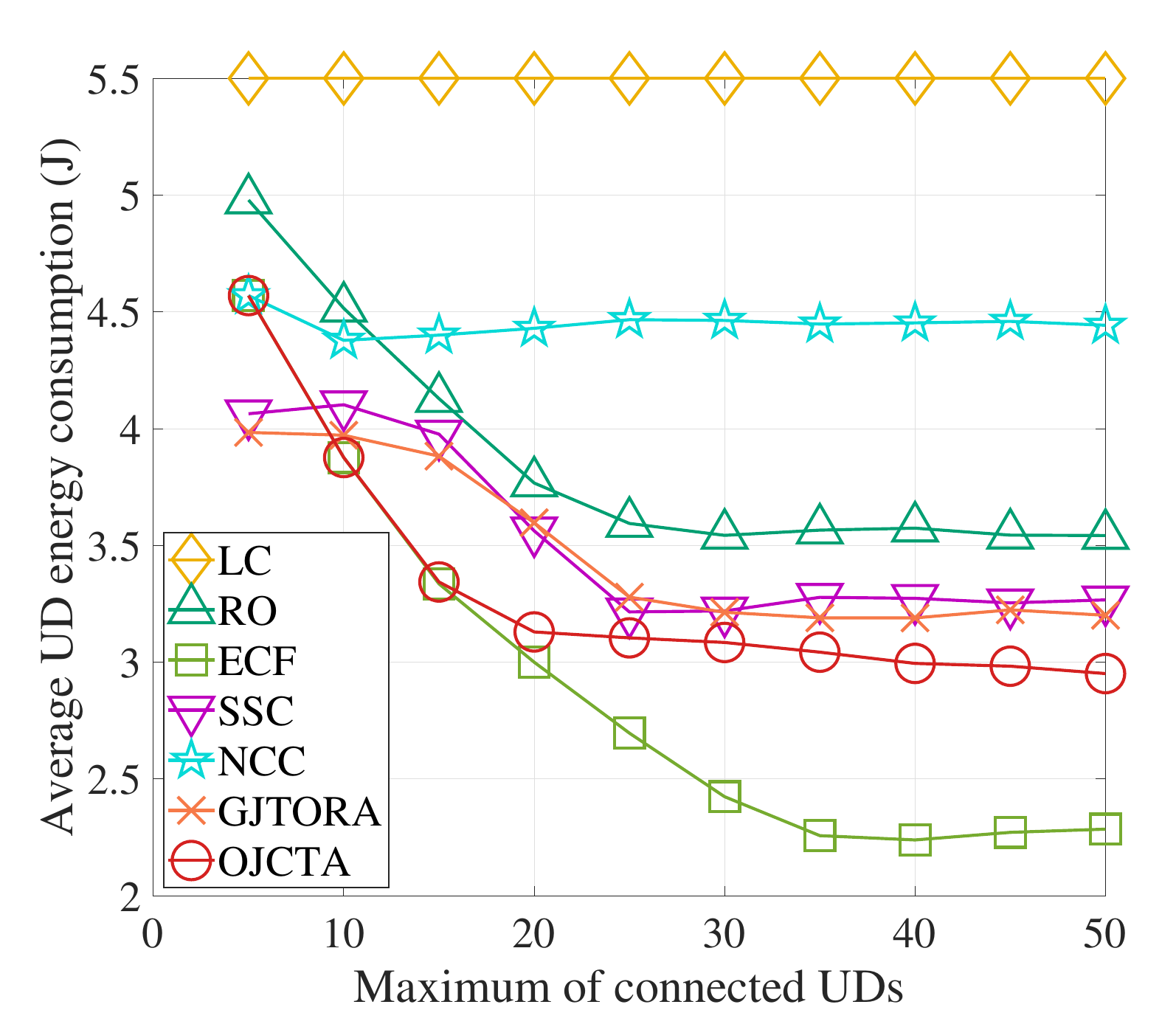}
		\end{minipage}
	}
	\subfigure[Average queuing delay of $Q^{\text{L}}$]
	{
             \begin{minipage}[t]{0.31\linewidth}
			\centering
			\includegraphics[scale=0.23]{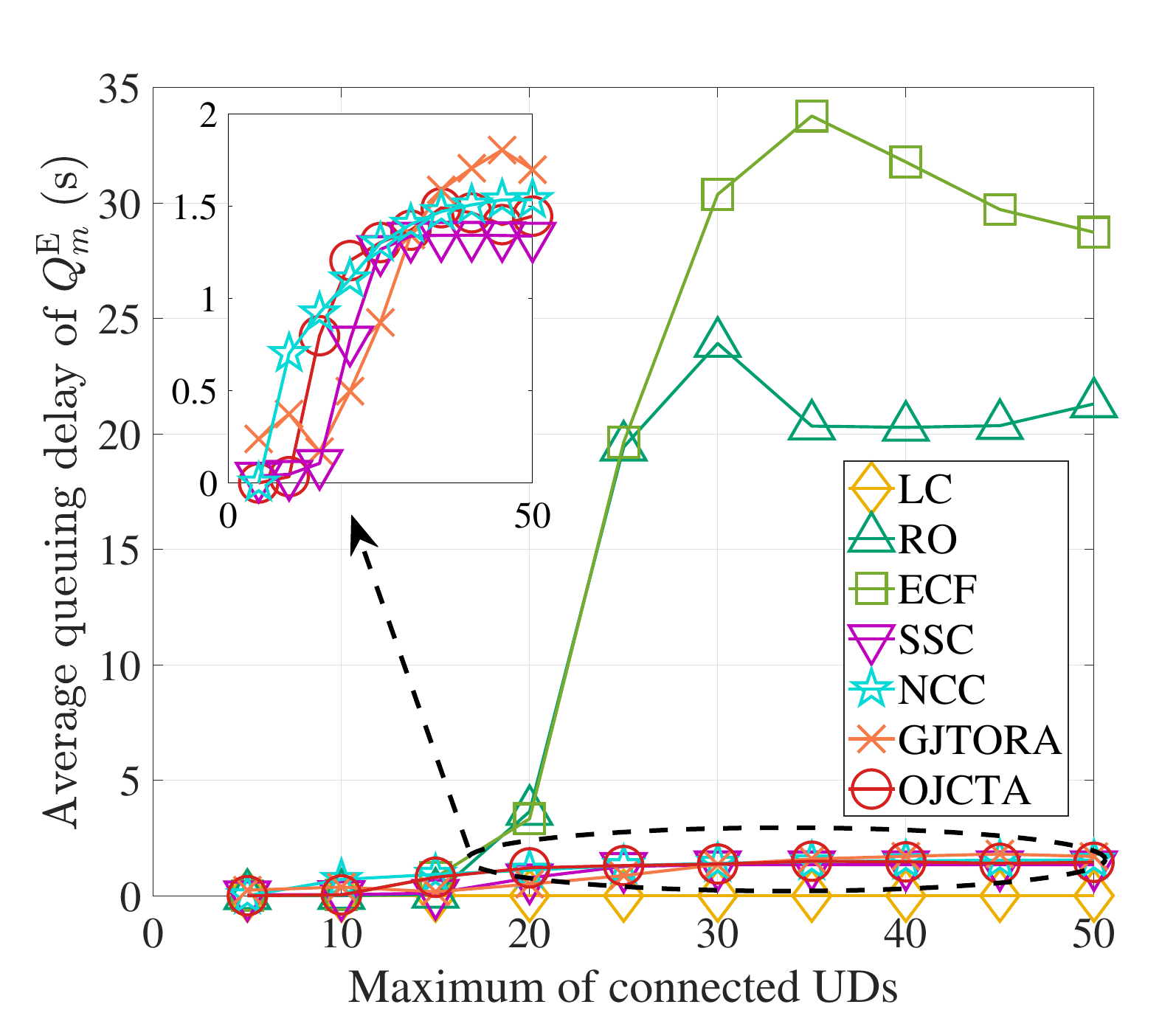}
		\end{minipage}
	}
	\subfigure[Average queuing delay of $Q^{\text{O}}$]
	{
             \begin{minipage}[t]{0.31\linewidth}
			\centering
			\includegraphics[scale=0.23]{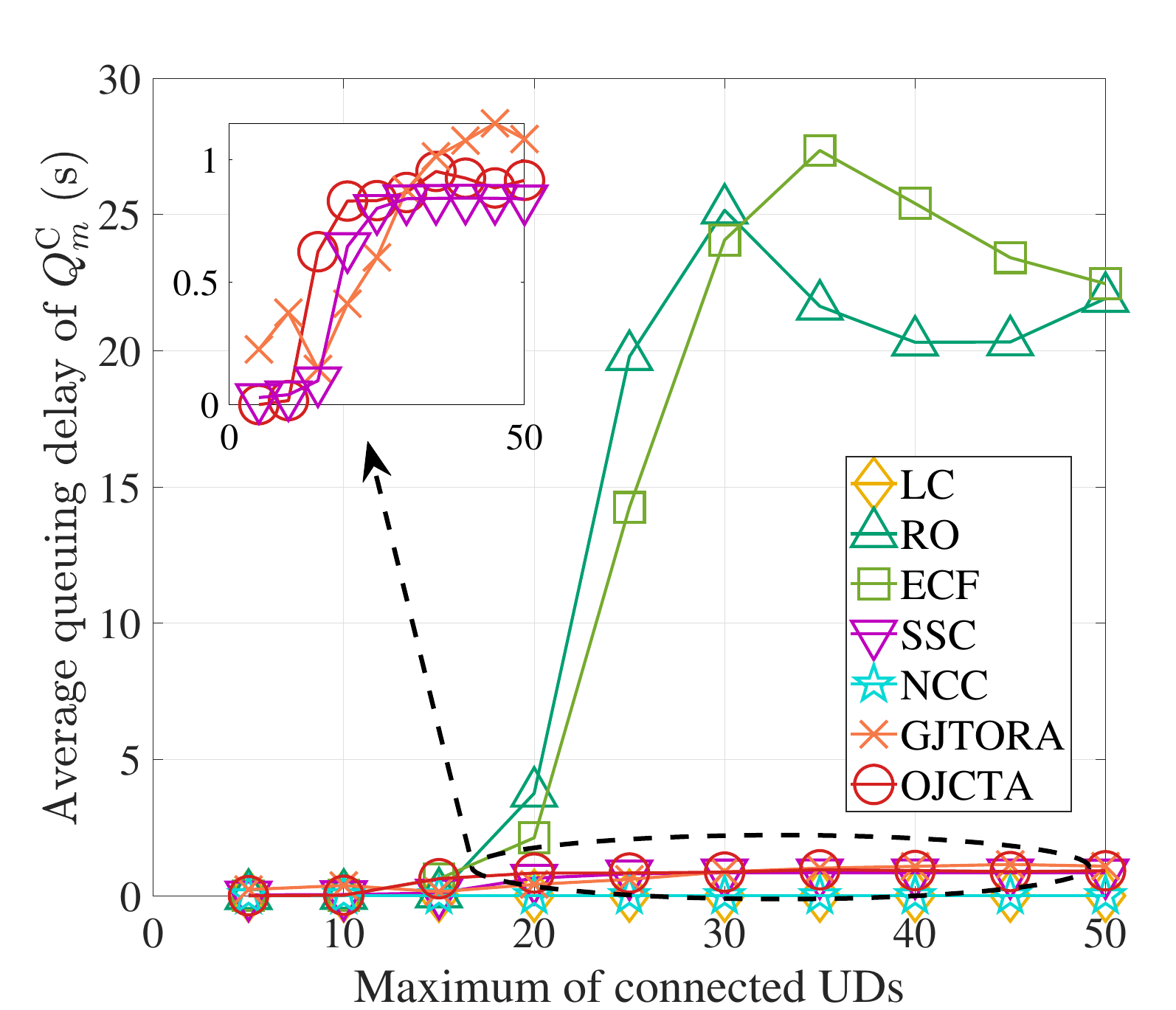}
		\end{minipage}
	}
	\centering
	\caption{System performance with the MEC server connection capacity.}
	\label{fig_s4}
	\vspace{-1em}
\end{figure*}

  \begin{figure*}[!hbt] 
	\centering
        \setlength{\abovecaptionskip}{2pt}%
	\setlength{\belowcaptionskip}{2pt}%
	\subfigure[Average UD energy consumption]
	{
              \begin{minipage}[t]{0.31\linewidth}
			\centering
			\includegraphics[scale=0.23]{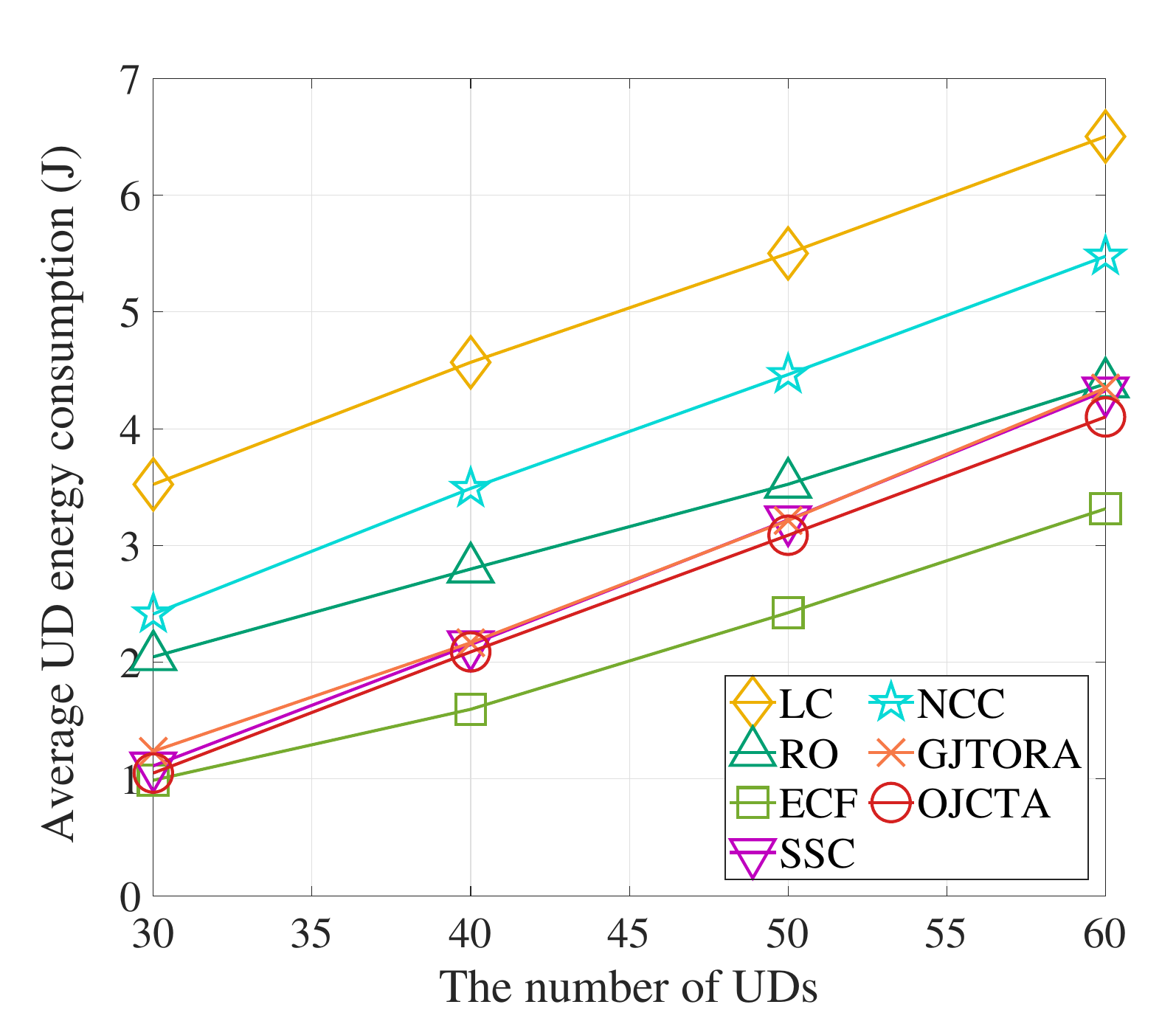}
		\end{minipage}
	}
    \subfigure[Average queuing delay of $Q^{\text{L}}$]
    {
        \begin{minipage}[t]{0.31\linewidth}
        \centering
        \includegraphics[scale=0.23]{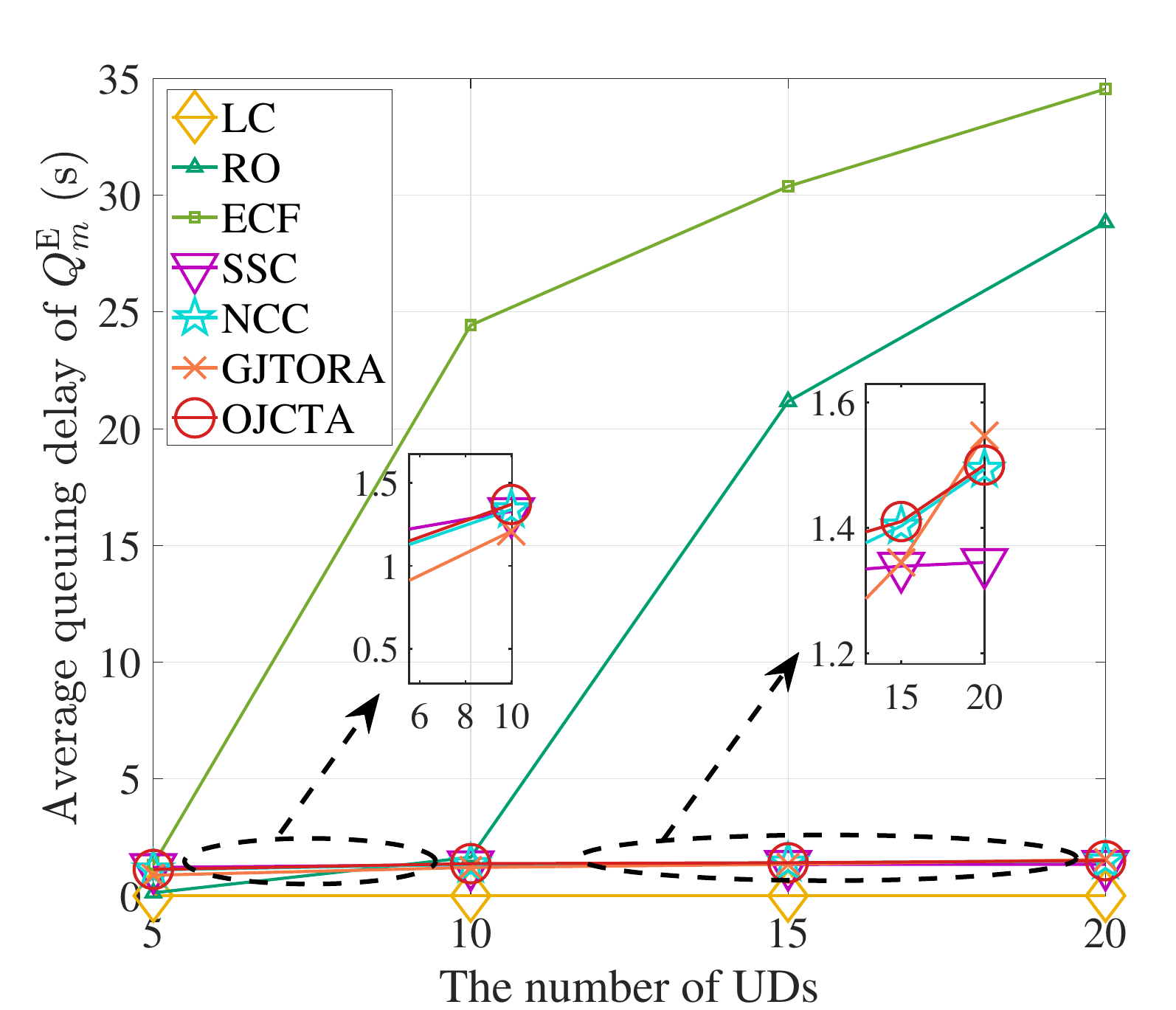}
       \end{minipage}
    }
    \subfigure[Average queuing delay of $Q^{\text{O}}$]
    {
        \begin{minipage}[t]{0.31\linewidth}
        \centering
        \includegraphics[scale=0.23]{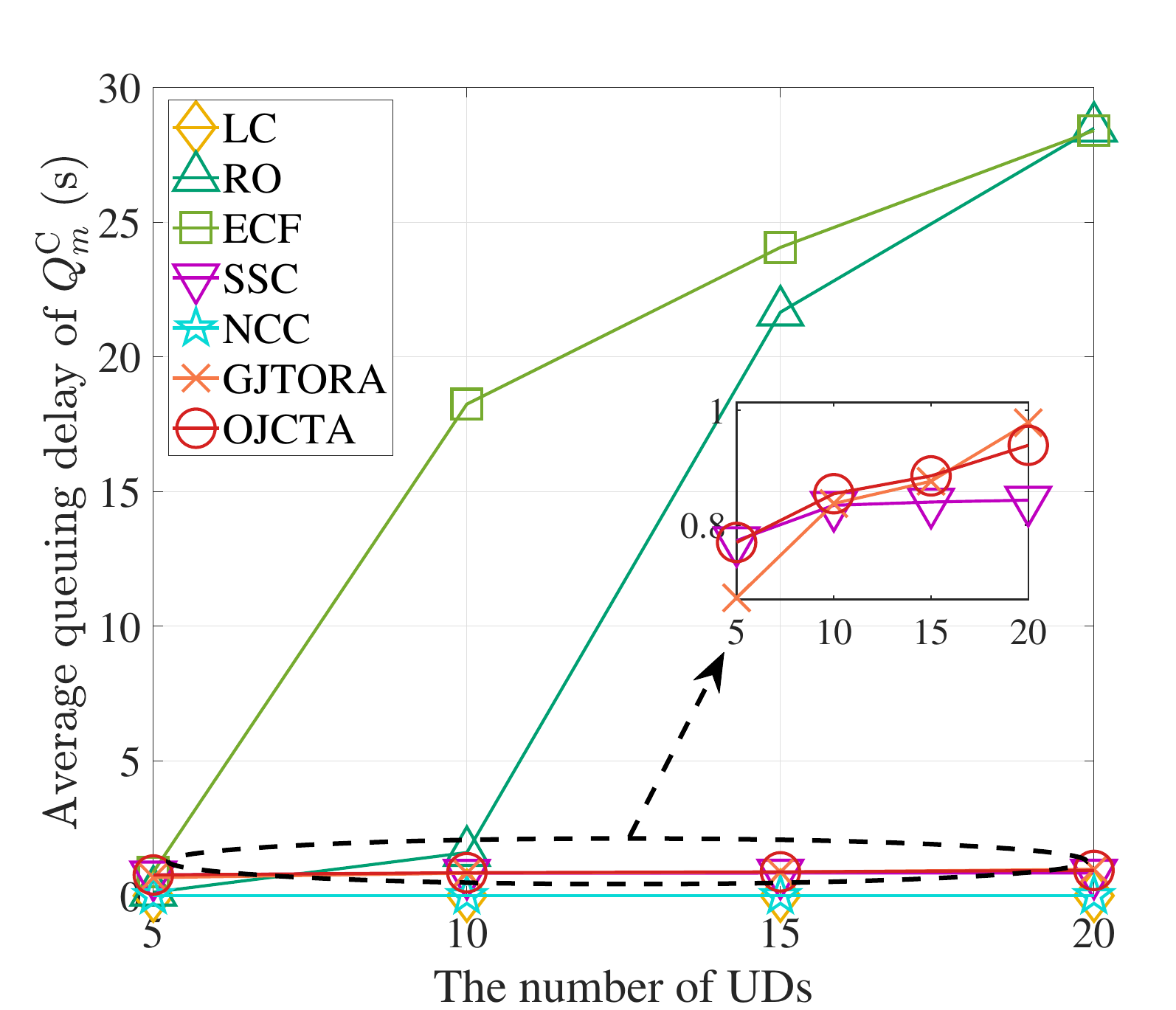}
       \end{minipage}
    }
    \caption{System performance with the number of UDs.}
	\label{fig_s5}
 \vspace{-1em}
\end{figure*}

\par \textbf{Impact of User Number.} Figs. \ref{fig_s5}(a), \ref{fig_s5}(b) and \ref{fig_s5}(c) show the impact of UD number on the average UD energy consumption, average queuing delay for edge computing queue, and average queuing delay, respectively.

\par Fig. \ref{fig_s5}(a) reveals that the average UD energy consumption rises across all approaches as the number of UDs grows. This is expected since a greater number of UDs results in more tasks to be processed. Moreover, with an increasing number of MDs, the OJCTA obtains lower UD energy consumption compared to LC, RO, SSC, NCC, and GJTORA, respectively reducing the average UD energy consumption by 36.94\%, 6.45\%, 5.17\%, 25.13\%, and 5.69\% in the relative dense scenario ($U\geq 50$). This can be attributed to the following reasons. First, the edge-cloud collaborative MEC architecture makes the  OJCTA more adaptive to the dense scenarios since the workload at the edge can be flexibly relived. Moreover, the proposed OJCTA can effectively minimize UD energy consumption by making online task offloading and resource allocation decisions based on the situated scenario. Besides, the bilateral matching-based task offloading method can effectively associate the different UDs with diverse processing requirements to the suitable MEC servers with different computing capabilities. Note that although ECF achieves the lowest UD energy consumption since it aims to minimize the energy consumption, it has the inferior performance in queuing delays.

\par Figs. \ref{fig_s5}(b) and \ref{fig_s5}(c) show that the average queuing delays of the seven approaches also increase as the number of UDs grows. It is expected that a larger number of UDs introduces higher task offloading demands, which can lead to congestion at the MEC servers. Specifically, it can be observed that the approaches such as RO and ECF experience significant queuing delays due to their lack of long-term queuing delay constraints. In contrast, as the number of UDs increases, OJCTA maintains relatively stable and lower queuing delays. The reason is that the proposed OJCTA can effectively minimize the UD energy consumption under the constraints of long-term queuing delay by performing real-time decisions based on varying UD densities. Additionally, as we mentioned earlier, the edge-cloud collaboration of the OJCTA helps distribute the workload efficiently, OJCTA leverages edge-cloud collaboration to efficiently distribute the workload, preventing bottlenecks and maintaining queuing stability as the number of UDs grows.

\par In conclusion, the proposed OJCTA exhibit superior scalability to adapt to the environment with varying densities. Moreover, OJCTA can efficiently accomplish the tasks with lower UD energy consumption while maintaining lower queuing delays and stable queuing performance.

%
%

\section{Conclusion}
\label{sec_conclusion}

In this work, we have studied an online collaborative communication resource allocation and task offloading for MEC system. First, we have designed an edge-cloud collaborative MEC architecture, where the MEC servers and the cloud server collaboratively provide offloading services for UDs. Moreover, we have formulated the EEDAOP to minimize the UD energy consumption under the constraints of task deadlines and long-term queuing delays. Furthermore, we have proposed an OJCTA to solve the formulated optimization problem. Specifically, the future-dependent EEDAOP has been first transformed into an online problem. Then, a two-stage alternating optimization method was presented for online task offloading and resource allocation. The simulation results have indicated that the proposed OJCTA achieves superior performance with respect to UD energy consumption, while maintaining moderate and stable queuing delays for both cloud computing queue and edge computing queue. This indicates that OJCTA can effectively reduce UD energy consumption while constraining the queuing delays and ensuring queuing stability. Moreover, the proposed OJCTA can effectively adapt to varying bandwidth conditions and MEC connection capacities by dynamically adjusting the decisions of task offloading and resource allocation. Besides, OJCTA have demonstrated superior scalability and stable queuing performance as the number of UDs increases.

\ifCLASSOPTIONcaptionsoff
\newpage
\fi

\bibliographystyle{IEEEtran}
\bibliography{references.bib}

\begin{IEEEbiography}[{\includegraphics[width=1in,height=1.23in,clip,keepaspectratio]{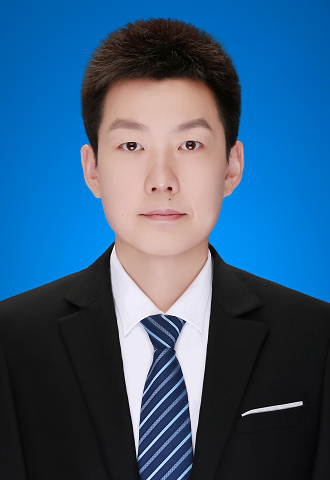}}]{Geng Sun} (S'17-M'19) received the B.S. degree in communication engineering from Dalian Polytechnic University, and the Ph.D. degree in computer science and technology from Jilin University, in 2011 and 2018, respectively. He was a Visiting Researcher with the School of Electrical and Computer Engineering, Georgia Institute of Technology, USA. He is an Associate Professor in College of Computer Science and Technology at Jilin University, and his research interests include wireless networks, UAV communications, collaborative beamforming and optimizations.
\end{IEEEbiography}


\begin{IEEEbiography}[{\includegraphics[width=1in,height=1.23in,clip,keepaspectratio]{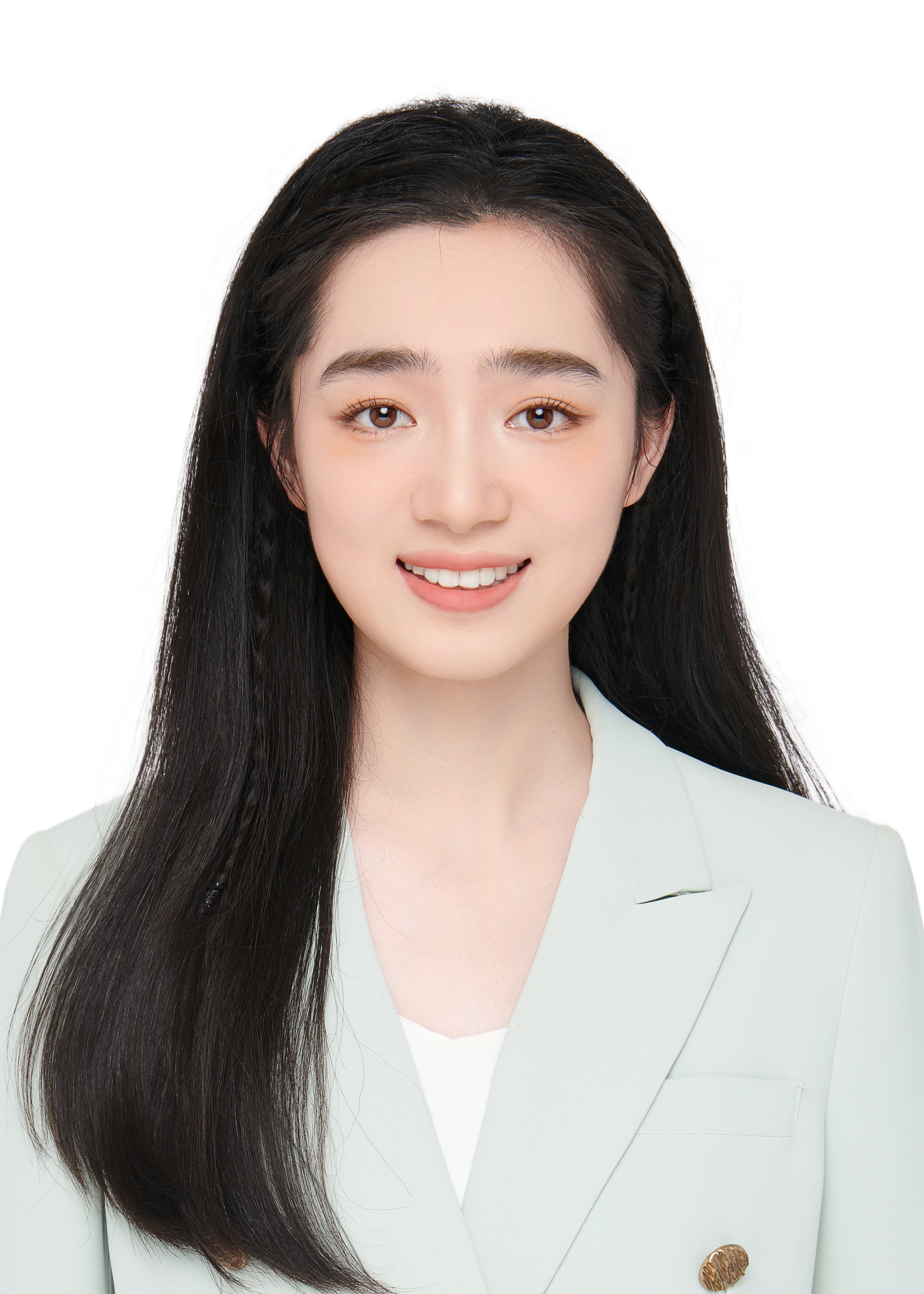}}]{Minghua Yuan} received a BS degree in Information Security from Wuhan University of Science and Technology, Wuhan, China, in 2022. She is currently working toward the MS degree in Software Engineering at Jilin University, Changchun, China. Her research interests include mobile edge computing and optimizations.
\end{IEEEbiography}

 \vspace{-3.5 em}

\begin{IEEEbiography}[{\includegraphics[width=1in,height=1.23in,clip,keepaspectratio]{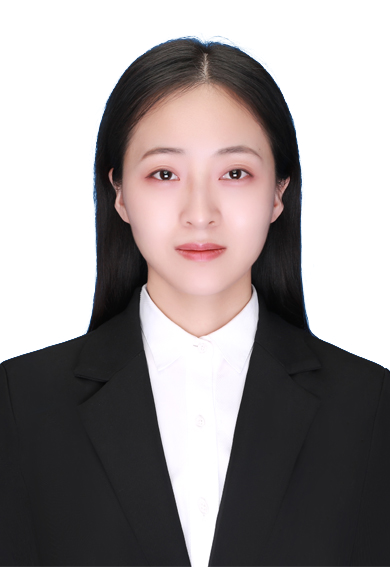}}]{Zemin Sun} received a BS degree in Software Engineering, an MS degree and a Ph.D degree in Computer Science and Technology from Jilin University, Changchun, China, in 2015, 2018, and 2022, respectively. Her research interests include mobile edge computing, UAV communications, and game theory. 
\end{IEEEbiography}

 \vspace{-3.5 em}
 
\begin{IEEEbiography}[{\includegraphics[width=1in,height=1.23in,clip,keepaspectratio]{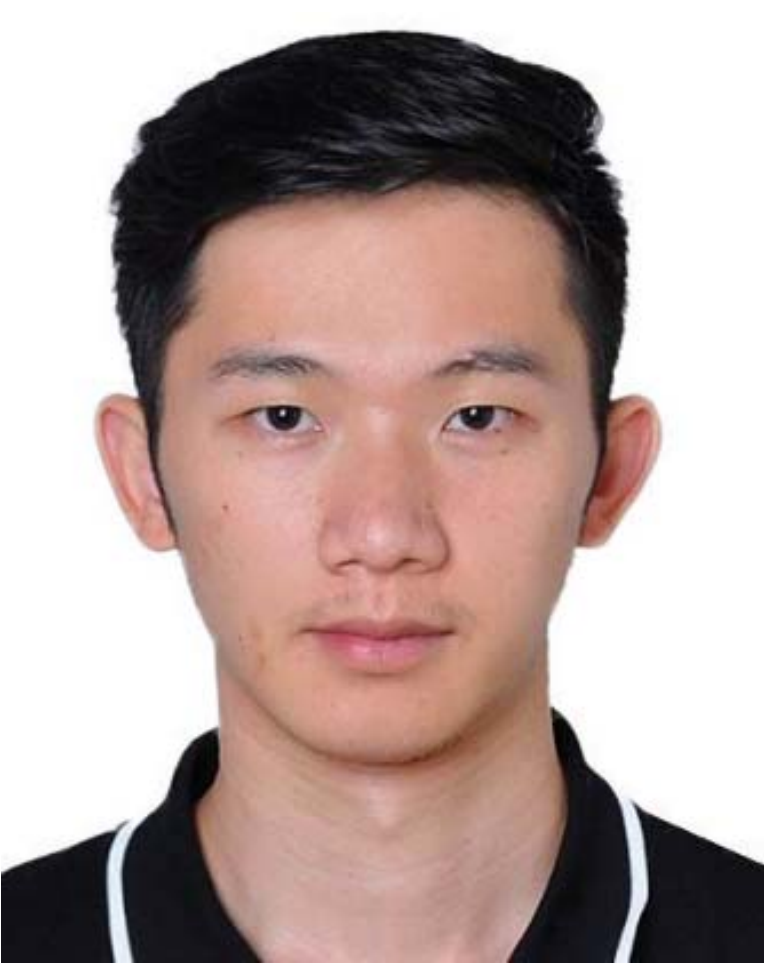}}]{Jiacheng Wang} received the Ph.D. degree from the School of Communication and Information Engineering, Chongqing University of Posts and Telecommunications, Chongqing, China. He is currently a Research Associate in computer science and
engineering with Nanyang Technological University,
Singapore. His research interests include wireless
sensing, semantic communications, and metaverse.
\end{IEEEbiography}

\vspace{-3.5 em}
 
\begin{IEEEbiography}[{\includegraphics[width=1in,height=1.23in,clip,keepaspectratio]{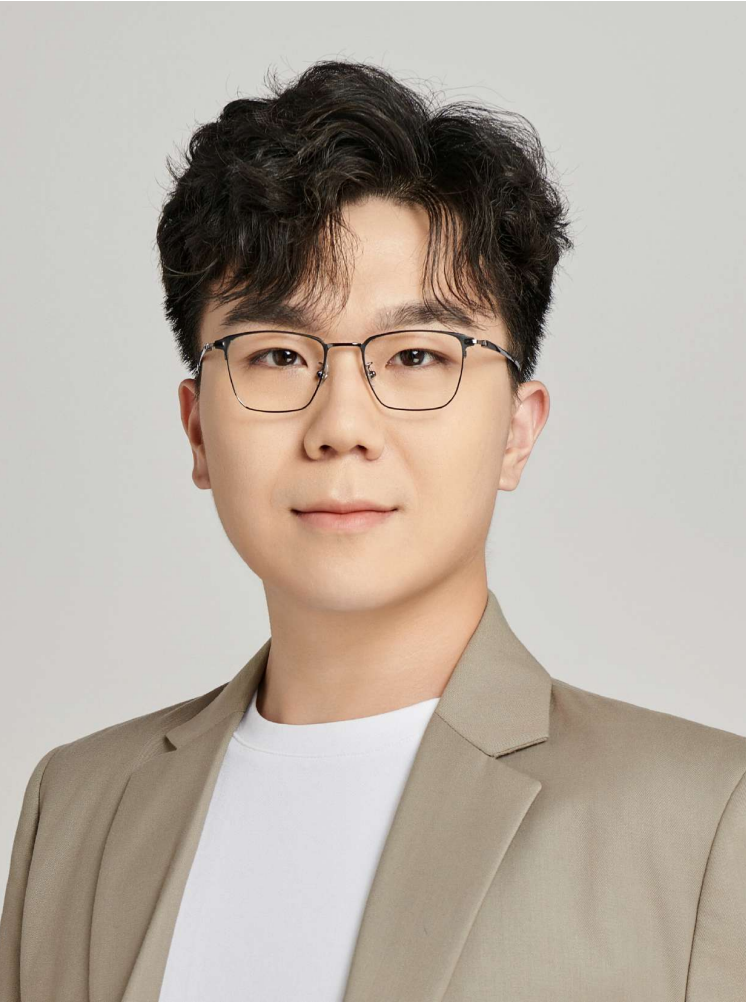}}]{Hongyang Du} is an assistant professor at the Department of Electrical and Electronic Engineering, The University of Hong Kong. He received the BEng degree from the School of Electronic and Information Engineering, Beijing Jiaotong University, Beijing, in 2021, and the PhD degree from the Interdisciplinary Graduate Program at the College of Computing and Data Science, Energy Research Institute @ NTU, Nanyang Technological University, Singapore, in 2024. He is the Editor-in-Chief assistant of IEEE Communications Surveys \& Tutorials (2022-2024). He is the recipient of the IEEE Daniel E. Noble Fellowship Award from the IEEE Vehicular Technology Society in 2022, the IEEE Signal Processing Society Scholarship from the IEEE Signal Processing Society in 2023, the Chinese Government Award for Outstanding Students Abroad in 2023, and the Singapore Data Science Consortium (SDSC) Dissertation Research Fellowship in 2023. He was recognized as an exemplary reviewer of the IEEE Transactions on Communications and IEEE Communications Letters in 2021. His research interests include edge intelligence, generative AI, semantic communications, and network management.
\end{IEEEbiography}

\vspace{-3.5 em}
\begin{IEEEbiography}[{\includegraphics[width=1in,height=1.23in,clip,keepaspectratio]{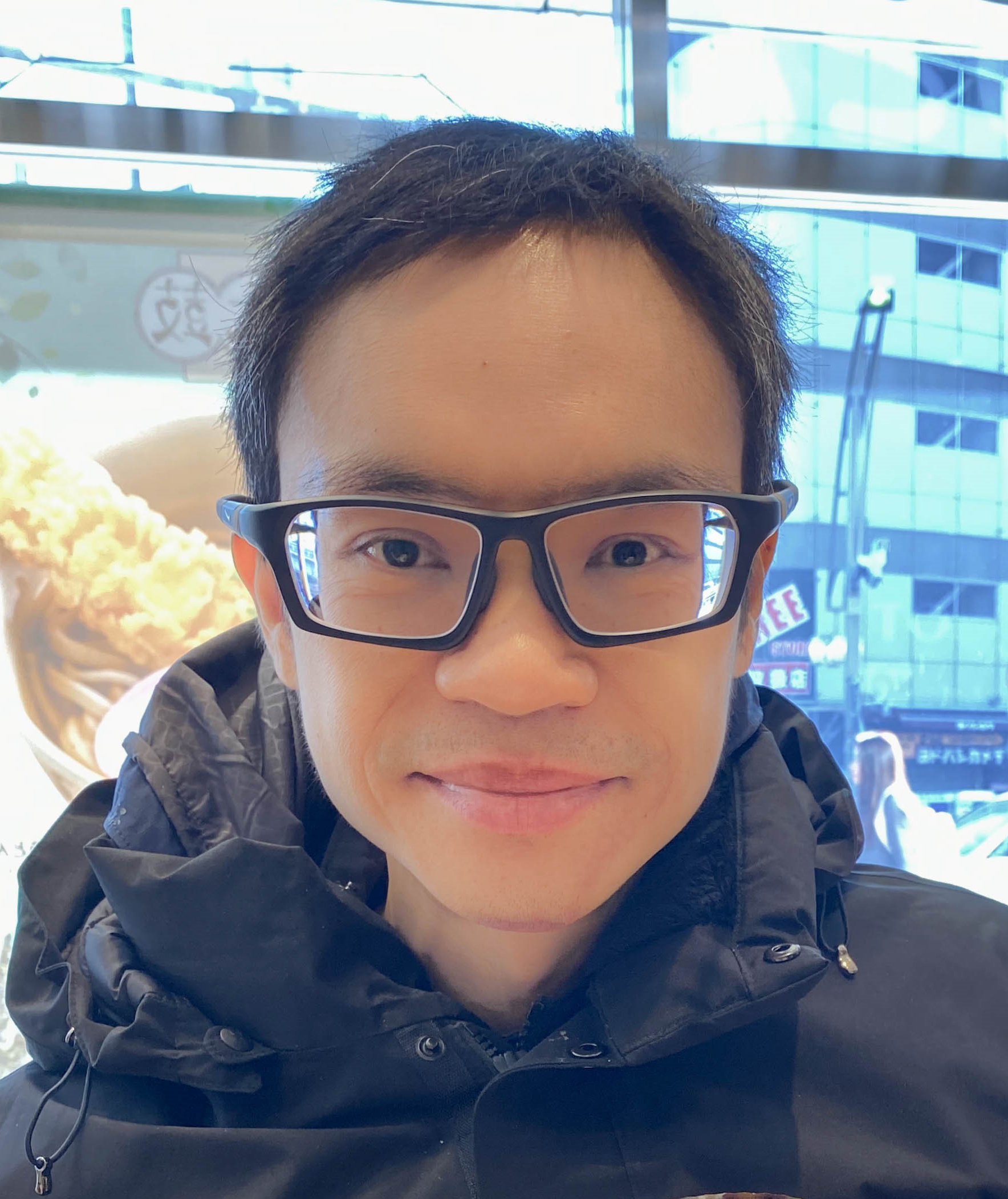}}]{Dusit Niyato} (Fellow, IEEE) received the B.Eng. degree from the King Mongkuts Institute of Technology Ladkrabang (KMITL), Thailand, in 1999, and the Ph.D. degree in electrical and computer engineering from the University of Manitoba, Canada, in 2008. He is currently a Professor with the School of Computer Science and Engineering, Nanyang Technological University, Singapore. His research interests include the Internet of Things (IoT), machine learning, and incentive mechanism design.  He won the IEEE Communications Society Best Survey Paper 1 Award, IEEE Asia-Pacific Board Outstanding Paper Award, the IEEE Computer Society Middle Career Researcher Award for Excellence in Scalable Computing and Distinguished Technical Achievement Recognition Award of IEEE ComSoc Technical Committee on Green Communications and Computing. He also won a number of best paper awards including IEEE Wireless Communications and Networking Conference, IEEE International Conference on Communications, IEEE ComSoc Communication Systems Integration and Modeling Technical Committee and IEEE ComSoc Signal Processing and Computing for Communications Technical Committee 2021. Currently, He is serving as Editor-in-Chief of IEEE Communications Surveys and Tutorials, an area editor of IEEE Transactions on Vehicular Technology, and editor of IEEE Transactions on Wireless Communications. 
\end{IEEEbiography}


\begin{IEEEbiography}[{\includegraphics[width=1in,height=1.23in,clip,keepaspectratio]{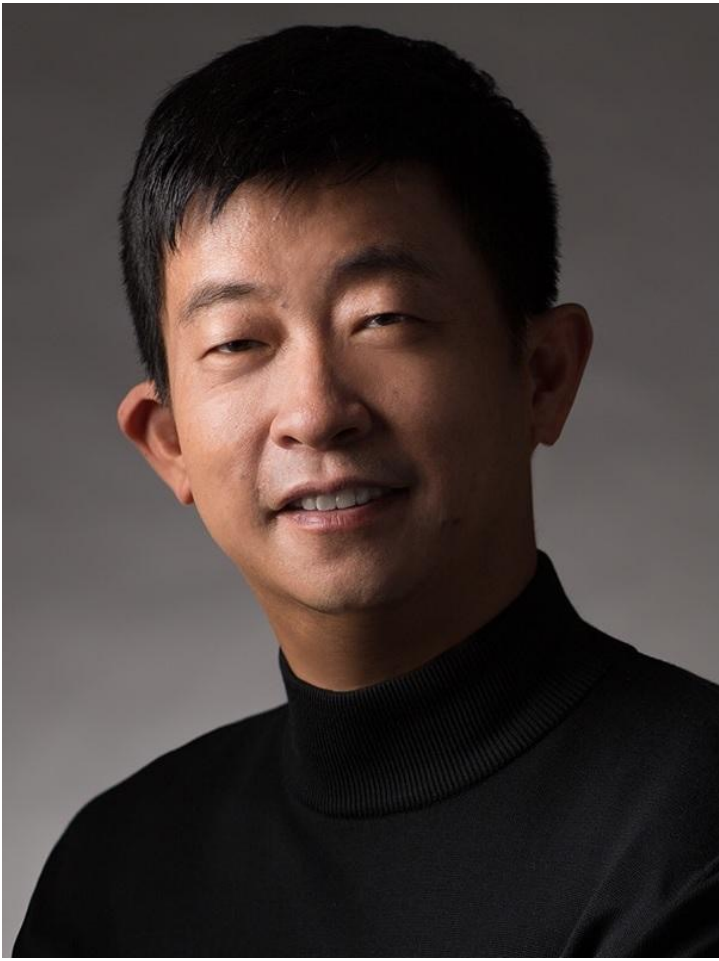}}]{Zhu Han} (Fellow, IEEE) received the B.S. degree in electronic engineering from Tsinghua University, in 1997, and the M.S. and Ph.D. degrees in electrical and computer engineering from the University of Maryland, College Park, in 1999 and 2003, respectively. From 2000 to 2002, he was an R\&D Engineer of JDSU, Germantown, Maryland. From 2003 to 2006, he was a Research Associate at the University of Maryland. From 2006 to 2008, he was an assistant professor at Boise State University, Idaho. Currently, he is a John and Rebecca Moores Professor in the Electrical and Computer Engineering Department as well as in the Computer Science Department at the University of Houston, Texas. Dr. Hans main research targets on the novel game-theory related concepts critical to enabling efficient and distributive use of wireless networks with limited resources. His other research interests include wireless resource allocation and management, wireless communications and networking, quantum computing, data science, smart grid, carbon neutralization, security and privacy. Dr. Han received an NSF Career Award in 2010, the Fred W. Ellersick Prize of the IEEE Communication Society in 2011, the EURASIP Best Paper Award for the Journal on Advances in Signal Processing in 2015, IEEE Leonard G. Abraham Prize in the field of Communications Systems (best paper award in IEEE JSAC) in 2016, IEEE Vehicular Technology Society 2022 Best Land Transportation Paper Award, and several best paper awards in IEEE conferences. Dr. Han was an IEEE Communications Society Distinguished Lecturer from 2015 to 2018 and ACM Distinguished Speaker from 2022 to 2025, AAAS fellow since 2019, and ACM Fellow since 2024. Dr. Han is a 1\% highly cited researcher since 2017 according to Web of Science. Dr. Han is also the winner of the 2021 IEEE Kiyo Tomiyasu Award (an IEEE Field Award), for outstanding early to mid-career contributions to technologies holding the promise of innovative applications, with the following citation: “for contributions to game theory and distributed management of autonomous communication networks.” design.
\end{IEEEbiography}

\vspace{9 em}
\begin{IEEEbiography}[{\includegraphics[width=1in,height=1.23in,clip,keepaspectratio]{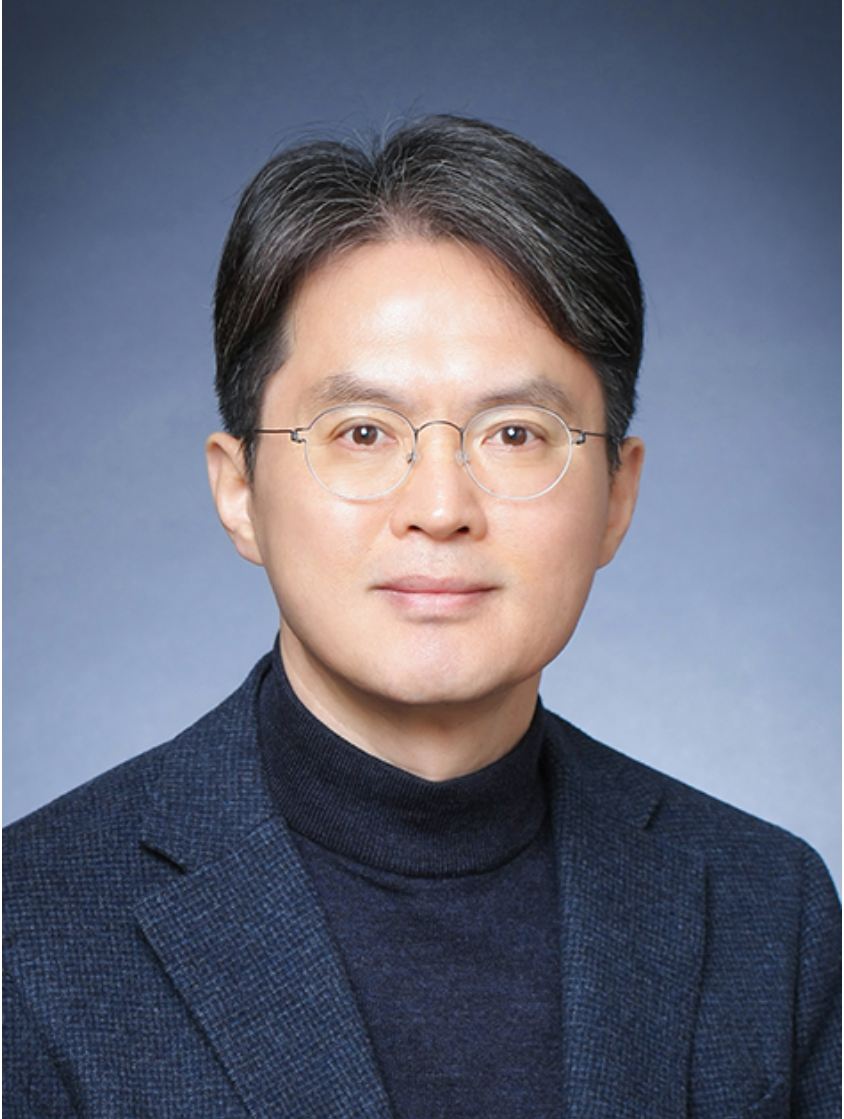}}]{Dong In Kim} (Fellow, IEEE) received the Ph.D. degree in electrical engineering from the University of Southern California, Los Angeles, CA, USA, in 1990. He was a Tenured Professor with the School of Engineering Science, Simon Fraser University, Burnaby, BC, Canada. He is currently a Distinguished Professor with the College of Information and Communication Engineering, Sungkyunkwan University, Suwon, South Korea. He is a Fellow of the Korean Academy of Science and Technology and a Member of the National Academy of Engineering of Korea. He was the first recipient of the NRF of Korea Engineering Research Center in Wireless Communications for RF Energy Harvesting from 2014 to 2021. He received several research awards, including the 2023 IEEE ComSoc Best Survey Paper Award and the 2022 IEEE Best Land Transportation Paper Award. He was selected the 2019 recipient of the IEEE ComSoc Joseph LoCicero Award for Exemplary Service to Publications. He was the General Chair of the IEEE ICC 2022, Seoul. Since 2001, he has been serving as an Editor, an Editor at Large, and an Area Editor of Wireless Communications I for IEEE Transactions on Communications. From 2002 to 2011, he served as an Editor and a Founding Area Editor of Cross-Layer Design and Optimization for IEEE Transactions on Wireless Communications. From 2008 to 2011, he served as the Co-Editor- in-Chief for the IEEE/KICS Journal of Communications and Networks. He served as the Founding Editorin-Chief for the IEEE Wireless Communications Letters from 2012 to 2015. He has been listed as a 2020/2022 Highly Cited Researcher by Clarivate Analytics.
\end{IEEEbiography}

\end{document}